\documentclass[a4paper,UKenglish,cleveref, autoref, thm-restate]{lipics-v2021}
\pdfoutput=1 %uncomment to ensure pdflatex processing (mandatatory e.g. to submit to arXiv)
\title{Flattability of Priority Vector Addition Systems} %TODO Please add

\author{Roland Guttenberg}{Technical University of Munich, Germany}{guttenbe@in.tum.de}{https://orcid.org/0000-0001-6140-6707}{}

\authorrunning{Roland Guttenberg} %TODO mandatory. First: Use abbreviated first/middle names. Second (only in severe cases): Use first author plus 'et al.'

\Copyright{Roland Guttenberg} %TODO mandatory, please use full first names. LIPIcs license is "CC-BY";  http://creativecommons.org/licenses/by/3.0/

 %TODO mandatory: Please choose ACM 2012 classifications from https://dl.acm.org/ccs/ccs_flat.cfm 

 %TODO mandatory; please add comma-separated list of keywords

\category{Track B: Automata, Logic, Semantics, and Theory of Programming}

\acknowledgements{I thank my PhD advisor Javier Esparza for reading a first draft and providing feedback. I thank the anonymous reviewers for their helpful feedback.}%optional

\nolinenumbers %uncomment to disable line numbering

\EventEditors{Karl Bringmann, Martin Grohe, Gabriele Puppis, and Ola Svensson}
\EventNoEds{4}
\EventLongTitle{51st International Colloquium on Automata, Languages, and Programming (ICALP 2024)}
\EventShortTitle{ICALP 2024}
\EventAcronym{ICALP}
\EventYear{2024}
\EventDate{July 8--12, 2024}
\EventLocation{Tallinn, Estonia}
\EventLogo{}
\SeriesVolume{297}
\ArticleNo{134}
\usepackage{tikz}
\usetikzlibrary{arrows,calc,fit,shapes,automata,backgrounds,positioning}

\usepackage{pgfplots}

\pgfplotsset{width=5cm,compat=1.10}
\usepgfplotslibrary{fillbetween}

\usepackage{xcolor}
\usepackage{stmaryrd}
\usepackage{tikz}
\usetikzlibrary{automata, positioning, arrows, petri, backgrounds}

\newcommand{\N}{\mathbb{N}}

\newcommand{\Z}{\mathbb{Z}}

\newcommand{\Q}{\mathbb{Q}}

\newcommand{\UpwardClosure}[1]{\lceil#1\rceil}
\definecolor{niceredbright}{HTML}{bd0310}
\definecolor{nicebluebright}{HTML}{197b9b}
\definecolor{nicered}{HTML}{7f0a13}
\definecolor{niceblue}{HTML}{104354}
\definecolor{nicegreen}{HTML}{217516}
\definecolor{nicepurple}{HTML}{884bab}
\definecolor{nicebg}{HTML}{f6f0e4}
\definecolor{niceredlight}{HTML}{c9888d}
\definecolor{nicebluelight}{HTML}{78a4b8}
\definecolor{nicegreenlight}{HTML}{76de68}
\definecolor{nicepurplelight}{HTML}{bc87db}

\newcommand{\vect}[1]{\mathbf{#1}}

\newcommand{\vectSet}[1]{\mathbf{#1}}
\newcommand{\CanReach}[1]{\to_{#1}^{\ast}}
\DeclareMathOperator{\VAS}{\mathcal{V}}
\DeclareMathOperator{\FO}{FO}

\DeclareMathOperator{\dir}{dir}
\DeclareMathOperator{\dirOfRun}{ends}

\DeclareMathOperator{\target}{tgt}
\DeclareMathOperator{\source}{src}

\DeclareMathOperator{\Rel}{Rel}
\DeclareMathOperator{\dimD}{d}
\DeclareMathOperator{\mtc}{mtc}

\DeclareMathOperator{\IncX}{x++}
\DeclareMathOperator{\IncY}{y++}
\DeclareMathOperator{\IncZ}{z++}
\newcommand{\Expression}[0]{\vectSet{E}}
\DeclareMathOperator{\rank}{rank}
\DeclareMathOperator{\Component}{Comp}

\begin{document}

\maketitle

\begin{abstract}
Vector addition systems (VAS), also known as Petri nets, are a popular model of concurrent systems. Many problems from many areas reduce to the
\emph{reachability problem} for VAS, which consists of deciding whether a target configuration of a VAS is reachable from a given initial configuration. One of the main approaches to solve the problem on practical instances is called \emph{flattening}, intuitively removing nested loops. This technique is known to terminate for semilinear VAS due to \cite{Leroux13}. In this paper, we prove that also for VAS with nested zero tests, called Priority VAS, flattening does in fact terminate for all semilinear reachability relations. Furthermore, we prove that Priority VAS admit semilinear inductive invariants. Both of these results are obtained by defining a well-quasi-order on runs of Priority VAS which has good pumping properties.
\end{abstract}

\section{Introduction}\label{SectionIntroduction}

% !TeX root = Main.tex

Vector addition systems (VAS), also known as Petri nets, are a popular model of concurrent systems. VAS have a very rich theory and have been intensely studied. In particular, the \emph{reachability problem} for VAS, which consists of deciding whether a target configuration of a VAS is reachable from a given initial configuration, has been studied for over 50 years. It was proved decidable in the 1980s \cite{Mayr81,Kosaraju82, Lambert92}, but its complexity (Ackermann-complete) could only be determined recently \cite{CzerwinskiLLLM19, CzerwinskiO21, Leroux21}. 

In \cite{Leroux09} and  \cite{Leroux13}, Leroux proved two fundamental results about the reachability sets of VAS.
In \cite{Leroux09}, he showed that every configuration outside the reachability set $\vect{R}$ of a VAS is separated from $\vect{R}$ by a semilinear inductive invariant (for basic facts on semilinear sets see e.g. \cite{Haase18}). This immediately led to a very simple algorithm for the reachability problem consisting of two semi-algorithms, one enumerating all possible paths to certify reachability, and one enumerating all semilinear sets and checking if they are separating inductive invariants. 

In \cite{Leroux13}, he proved that if the reachability set of a VAS is semilinear, then it is \emph{flattable}. Flattability states the existence of a finite sequence
 \(\rho_1, \dots, \rho_r\) of transition sequences such that every reachable vector can be reached via a sequence in \(\rho_1^{\ast} \dots \rho_r^{\ast}\), i.e., by means of a ``flat'' expression without nested loops. Flattability leads to an algorithm for deciding whether a semilinear set is included in or equal to the reachability set of a given VAS, i.e. whether a VAS has the set of desired behaviours. If it is not included, guess the violating configuration and check it is unreachable, otherwise guess a linear path scheme and verify it. 
  
One major branch of ongoing research in the theory of VAS studies whether results like the above extend to more general systems \cite{Reinhardt08, Bonnet11, Bonnet12, RosaVelardoF11, AtigG11, LerouxPS14, LerouxST15, HofmanLLLST16, LazicS16, FinkelLS18, LerouxS20, BlondinL23}. In particular, the reachability problem has been proved decidable for Priority VAS, an extension of VAS in which counters can be tested for zero, albeit in restricted manner: there is a total order on the counters such that whenever a counter is tested for \(0\), all smaller counters are simultaneously tested as well. In a famous but very technical paper, Reinhardt proved that the reachability problem remains decidable for Priority VAS \cite{Reinhardt08}. In \cite{Bonnet11} and later in his thesis \cite{Bonnet12}, Bonnet presented a more accessible proof which was obtained by extending the result of \cite{Leroux09}, separability by inductive semilinear sets.

In this paper we extend the result of \cite{Leroux13} to arbitrary Priority VAS, and on the way obtain another proof that \cite{Leroux09} extends. That is, we show that 
1. Priority VAS admit semilinear inductive invariants, and  2. semilinear Priority VAS are flattable. Notice that 2. was not known even for the special case of one testable counter. Furthermore, as remarked in \cite{LerouxS20}, while two-dimensional vector addition systems with a zero test and a reset are effectively semilinear \cite{FinkelLS18}, they are not flattable in general. Hence, our second result establishes a theoretical limit of flattability. 

These results are obtained via two technical contributions of independent interest. 

\subparagraph{Regular expressions for Priority VAS.} We give a new characterization of the reachability relations of Priority VAS.
More precisely, we show that a relation is the reachability relation of a Priority VAS if and only if it can be represented as a regular expression over the reachability relations of standard VAS, with the restriction that the Kleene star operation can only be applied to monotone relations. For example in case of the Priority VASS in Figure \ref{FigureExampleVASS} as \(\VAS\), we would consider the VASS without the zero test transition as \(\VAS_0\), and if we are interested in the reachability relation starting at \(q_s\), ending at \(q_s\) and requiring counter x to start and end at \(0\), formally \(\rightarrow_{\VAS, q_s \to q_s}^{\ast} \cap \{x_{in}=x_{out}=0\}\), then we would rewrite \(\rightarrow_{\VAS, q_s \to q_s}^{\ast} \cap \{x_{in}=x_{out}=0\}=(\rightarrow_{\VAS_0, q_s \to q_t}^{\ast} \cap \{x_{in}=x_{out}=0\})^{\ast}\). I.e., instead we consider the inner normal VASS starting at \(q_s\), ending at \(q_t\) and fixing \(x\) to \(0\) at start and end, then taking the reflexive transitive closure of this relation. In general zero testing a coordinate will generate an expression of the form \(\Expression^{\ast}\), where \(\Expression\) fixes some coordinates to \(0\) at start and end. One important aspect of this characterization is that all intersections with linear relations (for example here with \(x_{in}=x_{out}=0\)) can be pushed purely to the inner VASS level, where they were dealt with in \cite{Leroux13}. Hence in our arguments we only have to consider how to deal with \(\circ\) and \(\ast\), not with intersections or projections.

\begin{figure}[h!]
\begin{minipage}{4.5cm}
	\begin{tikzpicture}[-, auto, node distance=0.5cm]
		\tikzset{every place/.append style={minimum size=0.4cm, nicebluebright}}
		\tikzset{every label/.append style={text=nicebluebright}}
		\tikzset{every transition/.style={minimum size=0.25cm}}
		\tikzset{every edge/.append style={font=\scriptsize}}
		
		\newcommand*{\distancesubx}{2.5cm}
	
		% States
		\node[place] (A) at (0.5,0) {\(q_s\)};
		\node[place] (B) at (0.5+0.5*\distancesubx, 0) {\(q_1\)};
		\node[place] (C) at (0.5+1.0*\distancesubx, 0) {\(q_2\)};
		\node[place, double] (D) at (0.5+1.5 * \distancesubx, 0) {\(q_t\)};
		\node[white!100] (E) at (0,0) {};
		
		% Edges
		\path[->, thick] (A) edge[] (B);
		\path[->, thick] (B) edge[] (C);
		\path[->, thick] (C) edge[] node[above] {\(\IncZ\)}(D);
		\path[->, thick, looseness=5] (B) edge[] node[above] {\(\IncX\)} (B);
		\path[->, thick, looseness=5] (C) edge[] node[above] {\(x--;\IncY\)} (C);
		\path[->, thick, in=90, out=90, looseness=1.5] (D) edge[] node[above] {\(x==0\)}(A);
		\path[->, thick] (E) edge[] (A);
		
	\end{tikzpicture}
\end{minipage}%
\begin{minipage}{10cm}
	\begin{tikzpicture}[-, auto, node distance=0.5cm]
		\tikzset{every place/.append style={minimum size=0.4cm, nicebluebright}}
		\tikzset{every label/.append style={text=nicebluebright}}
		\tikzset{every transition/.style={minimum size=0.25cm}}
		\tikzset{every edge/.append style={font=\scriptsize}}
		
		\newcommand*{\distancesubx}{2.5cm}
	
		% States
		\node[place] (A) at (0.5,0) {\(q_s\)};
		\node[place] (B) at (0.5+0.5*\distancesubx, 0) {\(q_1\)};
		\node[place] (C) at (0.5+1.0*\distancesubx, 0) {\(q_2\)};
		\node[place] (D) at (0.5+1.5 * \distancesubx, 0) {\(q_t\)};
		\node[white!100] (E) at (0,0) {};
		\node[place] (A2) at (0.5+2.1 * \distancesubx,0) {\(q_s'\)};
		\node[place] (B2) at (0.5+2.45*\distancesubx, 0) {\(q_1'\)};
		\node[place] (C2) at (0.5+2.95*\distancesubx, 0) {\(q_2'\)};
		\node[place, double] (D2) at (0.5+3.5 * \distancesubx, 0) {\(q_t'\)};
		
		% Edges
		\path[->, thick] (A2) edge[] (B2);
		\path[->, thick] (B2) edge[] (C2);
		\path[->, thick] (C2) edge[] node[above] {\(\IncZ\)} (D2);
		\path[->, thick, looseness=5] (B2) edge[] node[above] {\(\IncX\)} (B2);
		\path[->, thick, looseness=5] (C2) edge[] node[above] {\(x--;\IncY\)} (C2);
		
		\path[->, thick] (E) edge[] (A);
		
		\path[->, thick] (A) edge[] (B);
		\path[->, thick] (B) edge[] (C);
		\path[->, thick] (C) edge[] node[above] {\(\IncZ\)} (D);
		\path[->, thick, in=90, out=90, looseness=1.5] (D) edge[] node[above] {\(x==0\)}(A);
		
		\path[->, thick] (D) edge[] node[above] {\(x==0\)}(A2);
	\end{tikzpicture}
\end{minipage}
\caption{Example of a PVASS and an equivalent (for \(x_{in}=0\)) flattened version.} \label{FigureIntroductoryExample} \label{FigureExampleVASS}
\end{figure}
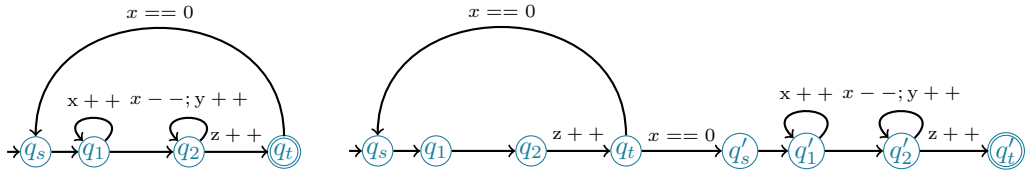

Characterizations of complicated relations using RegEx over simpler relations have already proven useful in other contexts \cite{PiskacK08, HaaseZ19}.

\subparagraph{A well-quasi-order (wqo) on the set of runs of Priority VAS.}  A wqo on a set is a partial order such that every subset has finitely many minimal elements. There exist wqos on many kinds of objects: vectors, sequences, trees, or graphs. A wqo on the set of runs of a PVAS provides the following decomposition: Let \(\source(\rho)\) denote the source of a run \(\rho\), \(\target(\rho)\) denote its target and \(\dirOfRun(\rho)=(\source(\rho), \target(\rho))\) its \emph{pair of ends}, i.e. source/target pair. Let \(\Omega_{min}\) be the finite set of minimal runs. Then the reachability relation \(\rightarrow_{\VAS}^{\ast}\) can be written as
\(\rightarrow^{\ast}=\bigcup_{\rho \in \Omega_{min}} \{\dirOfRun(\rho') \mid \rho' \geq \rho \},\)
i.e. we observe that every pair of configurations \(\vect{c}, \vect{c}'\) such that \(\vect{c} \rightarrow^{\ast} \vect{c}'\) is witnessed by some run, and then we split runs into one group for every minimal run. Intuitively, this reduces the problem of proving flattability of a VAS to proving flattability in each of the groups determined by the minimal runs. The proof that semilinear VAS are flattable follows this scheme.  More precisely, the proof, given in \cite{Leroux13}, takes the wqo on the runs of a VAS introduced by Jancar in \cite{Jancar90}, proves that it satisfies certain pumping properties, properties shown in \cite{Jancar90, Leroux11, Leroux12} and new ones, and derives the result. We proceed in the same way, starting from a wqo on the set of runs of a Priority VAS\footnote{The wqo for VAS does not respect the zero tests and hence does not work, a new ordering is necessary.}.

Let us now cconsider the concrete case of pumping and flattening in the example in Figure \ref{FigureExampleVASS}. Consider the run \(\rho: q_s \to q_1 \to q_2 \to q_t\) which does not use any of the self-loops. Intuitively, if we did not have the \(x_{in}=x_{out}=0\) restriction, both the \(x\) and \(y\) coordinates can be pumped arbitrarily along this run. In order to pump the \(x\) coordinate simply add one use of the self-loop on \(q_1\), and to increase the \(y\) coordinate add both self-loops once. Observe that despite using a loop which is not non-negative in the second case, and hence could not be arbitrarily often repeated by itself, this loop only decreases a coordinate which was already pumped prior in the run, i.e. the sequence of loops can be pumped. 

Using this basic image of pumping, we can now describe the idea of our proof. In the example, define \(\Expression:=\rightarrow_{\VAS_0, q_s \to q_t}^{\ast} \cap \{x_{in}=x_{out}=0\}\), then the target is \(\Expression^{\ast}\) as explained above. One first determines by induction hypothesis a decomposition of \(\Expression\) into groups of runs. In this case there is one group, as the run above is the unique minimal run.

Now we proceed to describe runs of \(\Expression^{\ast}\) as sequences of \emph{which group of \(\Expression\) was used}. I.e. we performed one important mental step here: A run of \(\Expression^{\ast}\) is now viewed as a sequence \(\vect{C}_0 \to_{\Expression} \vect{C}_1 \to_{\Expression} \dots \to_{\Expression} \vect{C}_r\) of steps in \(\Expression\), i.e. every step is now one application of the outer loop, and we \emph{abstract away} the information of how precisely these steps look, in particular how often the inner self-loops were taken. This leads to pumping a vector into a run having two cases: Either add more outer loops/transitions, the same as for VASSs, or \emph{increase one of the already existing loops}. For example for the run \(q_s(0,0,0) \to_{\Expression} q_t(0,0,1)\) described above, with respect to \(\Expression^{\ast}\), we can pump both the \(y\) and \(z\) coordinates arbitrarily. Pumping \(y\) is an instance of \emph{increasing an existing loop}: Change the existing loop by taking the inner self-loops. On the other hand, pumping \(z\) is the other type: We simply add more instances of the outer loop. The resulting PVASS without any nested loops is depicted in the right of  Figure \ref{FigureExampleVASS}. Observe in particular that for pumping those vectors it was not necessary to add arbitrarily many repetitions of the outer loop which take the inner loops arbitrarily often.

\subparagraph{Structure of the paper.}  In Section \ref{SectionSimplePreliminaries} we define a few preliminaries. Section \ref{SectionVAS} introduces VAS and Priority VAS.  Our first result, the characterization of the reachability relations of Priority VAS in terms of regular expressions, is proved in Section \ref{SectionRegExEquivalence}. In Section \ref{SectionWQO} we define well-quasi-orders, and in particular our novel wqo on runs of Priority VAS. Section \ref{SectionGeometricPreliminaries} introduces geometric preliminaries and previous results about VAS needed to state and prove our results. Section \ref{SectionProofFinalTheorem} defines flattability, and proves our main result.

\section{Preliminaries} \label{SectionSimplePreliminaries}

% !TeX root = Main.tex

We let $\N, \mathbb{Z}, \mathbb{Q}, \mathbb{Q}_{\geq 0}$ denote the sets of natural numbers containing \(0\), integers, and (non-negative) rational numbers. We use uppercase letters for sets/relations and boldface for vectors and sets/relations of vectors. For the \(i\)-th entry of a vector \(\vect{x} \in \Q^d\) we write \(\vect{x}(i)\).

Given sets \(\vectSet{X},\vectSet{Y} \subseteq \mathbb{Q}^d, Z \subseteq \mathbb{Q}\), we write \(\vectSet{X}+\vectSet{Y}:=\{\vect{x}+\vect{y} \mid \vect{x} \in \vectSet{X}, \vect{y} \in \vectSet{Y}\}\) for the Minkowski sum and \(Z \cdot \vectSet{X}:=\{\lambda \cdot \vect{x} \mid \lambda \in Z, \vect{x} \in \vect{X}\}\). By identifying elements \(\vect{x}\in \mathbb{Q}^d\) with \(\{\vect{x}\}\), we define \(\vect{x}+\vectSet{X}:=\{\vect{x}\}+\vectSet{X}\), and similarly \(\lambda \cdot \vectSet{X}:=\{\lambda\} \cdot \vectSet{X}\) for \(\lambda \in \mathbb{Q}\). %We denote by \(\vect{X}^C\) the complement of \(\vect{X}\). 

A set \(\vectSet{L} \subseteq \N^d\) is \emph{linear} if \(\vectSet{L}=\vect{b}+\N \vect{p}_1 + \dots + \N \vect{p}_r\) with \(\vect{b}, \vect{p}_1, \dots, \vect{p}_r \in \N^d\). A relation \(\vectSet{L} \subseteq \N^{d'} \times \N^{d''}\) is linear if it is linear when viewed as a set. A set/relation \(\vect{S}\) is \emph{semilinear} if it is a finite union of linear sets/relations. The semilinear sets/relations coincide with the sets/relations definable via formulas \(\varphi \in \FO(\mathbb{N}, +)\), also called Presburger Arithmetic.

Given relations \(\vectSet{R}_1 \subseteq \N^{d'} \times \N^{d_{mid}}\) and \(\vectSet{R}_2 \subseteq \N^{d_{mid}} \times \N^{d''}\), we write \(\vectSet{R}_1 \circ \vectSet{R}_2 =\{(\vect{v}, \vect{w}) \in \N^{d'} \times \N^{d''} \mid \exists \vect{x} \in \N^{d_{mid}}: (\vect{v}, \vect{x}) \in \vectSet{R}_1, (\vect{x}, \vect{w}) \in \vectSet{R}_2\}\) for composition. Given \(\vectSet{R} \subseteq \N^{d'} \times \N^{d'}\), we write \(\vectSet{R}^{\ast}\) for the reflexive and transitive closure (w.r.t. \(\circ\)).

Let \(j,d',d'' \in \N\) with \(j \leq d', d''\). A relation \(\vectSet{R} \subseteq \N^{d'} \times \N^{d''}\) is \emph{monotone in the \(j\)-th last coordinate} if for every \((\vect{x}, \vect{y}) \in \vectSet{R}\) we also have \((\vect{x}+\vect{e}_{d'+1-j}, \vect{y}+\vect{e}_{d''+1-j}) \in \vectSet{R}\), where \(\vect{e}_k\) is the \(k\)-th unit vector\footnote{The reason for starting to count coordinates from the end will be explained in the next section.}.  A relation \(\vectSet{R} \subseteq \N^{d'} \times \N^{d''}\) is \emph{monotone} if \(d'=d''\) and 
\(\vectSet{R}\) is monotone in every coordinate.

\section{Vector Addition Systems and Priority Vector Addition Systems} \label{SectionVAS}

% !TeX root = Main.tex

A \emph{priority vector addition system with states} (PVASS) \(\VAS\) of dimension \(d \in \N\) is a finite directed multigraph \((Q,E)\), whose edges \(e\) are labelled with a pair of a vector \(f(e)\in \Z^d\) and a number \(g(e)\in \{0,\dots,d\}\). The set of configurations of \(\VAS\) is \(Q \times \N^d\). An edge \(e=(p,p')\) with label \((f(e), g(e))\) induces a relation \(\to_{e}\) on configurations via \(\vect{c}=(q, \vect{x}) \to_{e} \vect{c}'=(q',\vect{x}')\) if and only if \(q=p, q'=p'\), \(\vect{x}(j)=0\) for all \(1 \leq j \leq g(e)\) and \(\vect{x}'=\vect{x}+f(e)\). Intuitively, the edge can only be used in state \(p\) to move to state \(p'\) and adds the vector \(f(e)\) to the current configuration. However, two conditions have to be fulfilled: We have to again arrive at a configuration \(\vect{c}'\) (i.e. \(\vect{x}'\) has to stay non-negative), and \(\vect{x}\) must be \(0\) on the first \(g(e)\) coordinates. We say that these coordinates are \emph{tested for 0}. Observe that contrary to Minsky machines, if a counter \(i\) is tested for \(0\), also all smaller counters \(j \leq i\) are tested for \(0\).

We write \(\to_{\VAS}=\bigcup_{e\in E} \to_{e}\) and let \(\to_{\VAS}^{\ast}\) denote its reflexive and transitive closure. A run of \(\VAS\) is a finite sequence \(\rho=(\vect{c}_0, \vect{c}_1, \dots, \vect{c}_k)\) of configurations such that \(\vect{c}_i \to_{\VAS} \vect{c}_{i+1}\) for all \(0 \leq i \leq k-1\). The \emph{source} of the run \(\rho\) is the configuration \(\source(\rho):=\vect{c}_0\), and the \emph{target} is \(\target(\rho):=\vect{c}_k\). The \emph{pair of ends} of \(\rho\) is \(\dirOfRun(\rho)=(\source(\rho), \target(\rho))\). A configuration \(\target\) is reachable from \(\source\) in \(\VAS\) if \( \source \to_{\VAS}^{\ast} \target\), or equivalently if there exists a run \(\rho\) with \(\dirOfRun(\rho)=(\source,\target)\).

%We define the reachability set of a PVASS \(\VAS\) for a given set of starting configurations \(\vectSet{S}\subseteq Q \times \N^d\) via \(\poststarVAS(\vectSet{S})=\{\target \in Q \times \N^d \mid \exists \source \in \vectSet{S}: \source \to_{\VAS}^{\ast} \target\}\). If the PVASS \(\VAS\) is clear, we often also write \(\poststar(\vectSet{S})\) and \(\to^{\ast}\) instead.
A \emph{priority vector addition system} (PVAS) \(\VAS\) is a PVASS with only one state, a \emph{vector addition system with states} (VASS) is a PVASS where \(g(e)=0\) for every edge \(e\), i.e. no counter is ever tested for \(0\). A VAS is a PVAS which is also a VASS.

Since the class of reachability relations of PVASS is lacking some important closure properties, and we do not want to distinguish between PVASS and PVAS all the time, we instead consider a larger class of sets (which coincides for the two models). Intuitively, not every run is accepting anymore, instead a run has to start in a given initial state \(p\) and end in a given final state \(q\), and certain counters have to start and/or end with fixed values. The idea is to then view the relation as subset of \(\N^{d'} \times \N^{d''}\) where \(d', d''\) are the number of input and respectively output counters which are not fixed.

%\begin{definition}
%Let \(Q\) be a finite set, \(d \in \N\). Let \(C_{blank}=(q, \vect{x}_{blank}) \in Q \times (\N \cup \{blank\})^d\) be a vector, which can potentially also contain a fixed special label \(blank\). The \emph{subconfiguration} defined by \(C_{blank}\) is \(\{(q', \vect{x}) \in Q \times \N^d \mid q'=q, \vect{x}(i)=\vect{x}_{blank}(i) \text{ if } \vect{x}_{blank} \neq blank\}\).
%
%A set \(\vectSet{L}\) is a \emph{subconfiguration} if there exists a vector \(C_{blank} \in Q \times (\N \cup \{blank\})^d\) such that \(\vectSet{L}\) is the subconfiguration defined by \(C_{blank}\).
%\end{definition}
%
%We can now define the extended notion of reachability relation.

\begin{definition}
\cite{ClementeCLP17} A relation \(\vect{X} \subseteq \N^{d'} \times \N^{d''}\) is a \emph{(P)VASS section} if there exists a (P)VASS \(\VAS\) of dimension \(d \geq d', d''\), states \(p,q\) and vectors \(\vect{b}_s \in \N^{d-d'}, \vect{b}_t \in \N^{d-d''}\) such that

\(\vectSet{X}=\{(\vect{x}, \vect{y})\in \N^{d'} \times \N^{d''} \mid (p,(\vect{b}_s, \vect{x})) \CanReach{\VAS} (q, (\vect{b}_t, \vect{y}))\}\).
\end{definition}

This is the reason for defining monotonicity counting from the end: The same counter has different indices as unit vector because of fixing a different number of coordinates.

We write the section defined by the PVASS \(\VAS\), the states \(p,q\) and the vectors \(\vect{b}_s, \vect{b}_t\) as \(\binom{\VAS,p,q}{\vect{b}_s, \vect{b}_t}\). If \(\VAS\) is a PVAS, then we leave away the unique states and only write \(\binom{\VAS}{\vect{b}_s, \vect{b}_t}\). The reason for this notation is that PVASS sections should be viewed as an intersection of two relations: The reachability relation with fixed source state and target state, and the linear relation defined by the fixed coordinates. In the notation we like to split these parts. 

We have three remarks on the definition of PVASS sections.

\begin{remark} \label{RemarkFixedCoordinates}
We fix coordinates starting from the first, i.e. the most often zero tested coordinates are fixed first. This does not restrict the class of PVAS sections.%For a given PVASS \(\VAS\) with counters \(\{1,\dots,d\}\) and initial/final states \(q_{init}, q_{final}\), one can construct a PVASS \(\VAS'\) with counters \(\{1,\dots,d,1', \dots, d'\}\) which first simulates \(\VAS\) on the non \(`\) counters and at some point decides to change from \(q_{final}\) to a new state in which it copies the values of the counters to the \(`\) versions (possibly permuted to guarantee that the coordinates fixed for the PVASS section become the frontmost \(`\) counters). The PVASS section is then defined by requiring all non \(`\) counters to be fixed to \(0\), and the rest to be fixed as before. Since only non \(`\) counters are ever tested for \(0\), this fulfills the assumptions.
\end{remark}

\begin{remark} \label{RemarkStates}
At the cost of increasing the dimension by \(3\), states are a special case of fixed never-zero-tested coordinates \cite{HopcroftP79}, hence (P)VASS sections can equivalently be defined by (P)VAS. Furthermore, similar to how zero tests in Minsky machines can be assumed to only change the state, we will always require that \(f(e)(j)=0\) for all \(j \leq g(e)\), i.e. any counter which is being zero tested is not updated. To obtain this assumption, simply move to an intermediate state from which you perform the additions afterwards. % Namely, if \(\VAS\) uses states \(\{q_1, \dots, q_n\}\), then one constructs \(\VAS'\) by adding four counters \(\{1',\dots, 4'\}\) and keeping the invariant \(C(1')+C(2')+C(3')+C(4')=n\). Then for every action \(\vect{a}\) of \(\VAS\) from \(q_i\) to \(q_j\), \(\VAS'\) has an action \(\vect{a}'\) with same effect on the common counters, but effect \(\vect{a}'(1')=-i, \vect{a}'(2')=-n+i, \vect{a}'(3')=j, \vect{a}'(4')=n-j\) on the new coordinates. This action can by the invariant only be executed if \(C(1')=i, C(2')=n-i, C(3')=0, C(4')=0\) (i.e. it verifies the state). Then there are additional actions \(\vect{a}_j\) for every \(j \in \{1,\dots,n\}\) which transfer the value back from the \(3', 4'\) counters to \(1',2'\). There is also a construction increasing dimension only by 3 \cite{HopcroftP79}.
\end{remark}

\begin{remark}
When using states, we can w.l.o.g. require \(\vect{b}_s=0^{d-d'}\) and \(\vect{b}_t=0^{d-d''}\), i.e. fixed coordinates are fixed to 0. However, since it is sometimes preferable to not use states, we allow general vectors \(\vect{b}_s, \vect{b}_t\) for the fixed coordinates.
\end{remark}

\section{Equivalence of PVASS and Regular Expressions over VASS} \label{SectionRegExEquivalence}

% !TeX root = Main.tex

Next we define the grammar which we will then prove to be equivalent to PVASS sections. Intuitively, one considers regular expressions where leaves/letters are VASS sections \(\vect{Y}\). Intermediate relations might have different input and output dimensions, hence non-terminals depend on the dimensions, and composition requires matching dimensions.
\begin{definition}
Consider the following grammar with non-terminals \(\Expression_{d',d''}\) for \(d', d'' \in \N\):
\begin{align*}
\Expression_{d',d''}&=\vectSet{Y}_{d',d''} \mid \Expression_{d',d_{mid}} \circ \Expression_{d_{mid},d''} \mid \Expression_{d',d''} \cup \Expression_{d',d''} \\
\Expression_{d',d'}&=\vectSet{Y}_{d',d'} \mid \Expression_{d',d_{mid}} \circ \Expression_{d_{mid},d'} \mid \Expression_{d',d'} \cup \Expression_{d',d'} \mid \Expression_{d',d'}^{\ast},
\end{align*}

where the \(\vectSet{Y}_{d',d''}\) are VASS sections \(\subseteq \N^{d'} \times \N^{d''}\). An expression \(\Expression_{d',d''}\) defines in a natural way a relation \(\Rel(\Expression_{d',d''}) \subseteq \N^{d'} \times \N^{d''}\), by interpreting \(\circ\) as composition of relations, \(\cup\) as union and \(\ast\) as reflexive transitive closure. We usually write \(\Expression\) instead of \(\Expression_{d',d''}\) when the dimensions \(in(\Expression)=d'\) and \(out(\Expression)=d''\) are clear.
\end{definition}

 Before we state the main theorem of this section, there are two things to note about this definition of the semantics. 1) \(\ast\) by definition adds reflexivity, however it only does so in the non-fixed counters. This was one goal of the definition allowing different dimensions of intermediate objects. 2) The semantics for composition however are not as intuitive as it might seem, see the following example.

\begin{example}
Let \(\VAS\) be the 1-dimensional VAS with two transitions, incrementing \(x\) and decrementing \(x\). Then its reachability relation is \(\rightarrow_{\VAS}^{\ast}=\N \times \N\). Consider \(\Rel(\binom{\VAS}{\epsilon,0} \circ \binom{\VAS}{1,\epsilon}) \subseteq \N \times \N\). We have \(\Rel(\binom{\VAS}{\epsilon,0})=\N \times \{\epsilon\}\) and \(\Rel(\binom{\VAS}{1,\epsilon})=\{\epsilon\} \times \N\), where we write \(\epsilon \in \N^0\) for the unique empty product. Despite the fixed coordinates \(0,1 \in \N\) not matching up, we obtain \(\Rel(\binom{\VAS}{\epsilon,0} \circ \binom{\VAS}{1,\epsilon})=\N \times \N\) by definition. This is due to the composition being only defined on the remaining, i.e. non-fixed coordinates.
\end{example}

We can now state the main theorem of this section.

\begin{restatable}{theorem}{TheoremEquivalenceRegExPVAS}
A relation \(\vectSet{X} \subseteq \N^{d'} \times \N^{d''}\) is a PVASS section iff \(\vectSet{X}=\Rel(\Expression)\) for some \(\Expression\). \label{TheoremEquivalenceRegExPVAS}
\end{restatable}

\begin{proof}
``\(\Rightarrow\)'': Let \(\vectSet{X}=\binom{\VAS}{\vect{b}_s, \vect{b}_t} \subseteq \N^{d'} \times \N^{d''}\) for a \(d\)-dimensional PVAS \(\VAS\) \emph{without states}, since they produce the same class of sections as in Remark \ref{RemarkStates}. We will prove by induction on \(k:=\max_{e \in E} g(e)\), i.e. the maximal zero-tested counter, that every PVAS section \(\vectSet{X}_k \subseteq \N^{d'} \times \N^{d''}\) has an equivalent expression \(\Expression_k\) in the grammar. In the base case \(k=0\), \(\VAS\) is actually a VAS, and hence \(\vectSet{Y}:=\vectSet{X}_k\) is an equivalent expression in the grammar.

Induction step: \(k-1 \to k\): For \(j \in \{k-1, k\}\) let \(E_j:=\{e \in E \mid g(e) \leq j\}\), in particular \(Z_k:=E_{k} \setminus E_{k-1}\) are the edges with \(g(e)=k\), i.e. testing counters \(1,\dots,k\). Let \(\VAS_j\) be the PVASS with edges \(E_j\) and the same labels \(f(e)\) and \(g(e)\) as \(\VAS\). By induction, for \(\VAS_{k-1}\) and any vectors \(\vect{b}_s, \vect{b}_t\) there exists \(\Expression_{k-1, \vect{b}_s, \vect{b}_t}\) with \(\Rel(\Expression_{k-1, \vect{b}_s, \vect{b}_t})=\binom{\VAS_{k-1}}{\vect{b}_s, \vect{b}_t}\).

Importantly, the semantics \(\to_{e}\) of a single action \(e\in E\), even if \(e\) performs a zero test, can be defined by a VASS. In particular for \(Z_k\) we define the following \(d\)-dimensional VASS \(\VAS_{Z_k}\) with \(|Z_k|+2\) states \(Q=\{q_{in}, q_{fin}\} \cup Z_k\) and \(2|Z_k|\) actions: For every \(e \in Z_k\), let \(e'=(q_{in}, e)\) with label \(f(e')=f(e)\) and \(e''=(e,q_{fin})\) with label \(f(e'')=\vect{0}\). Intuitively, we non-deterministically choose an \(e \in Z_k\) and execute its action, afterwards moving to \(q_{fin}\). We do not perform the zero test, instead this will be done using the VASS section. Namely we define \(\Expression_{Z_k}:=\binom{\VAS_{Z_k}, q_{in}, q_{fin}}{0^k, 0^k}\), i.e. we require the first k counters to be \(0\) at the start and end. That their values will still be \(0^k\) also at the end follows by the assumption in Remark \ref{RemarkStates}, that zero tests do not change the counters they are testing. Then we define

\(\Expression_{k, \vect{b}_s, \vect{b}_t}:=\Expression_{k-1,\vect{b}_s, \vect{b}_t} \cup \Expression_{k-1,\vect{b}_s, 0^k} \circ \left ( \Expression_{k-1,0^k, 0^k} \cup \Expression_{Z_k} \right)^{\ast} \circ \Expression_{k-1, 0^k, \vect{b}_t}.\)

Intuitively, the expression says the following: Either the zero testing actions in \(Z_k\) are never used, or we move from \(\vect{b}_s\) to a configuration with the first \(k\) counters fixed to \(0\), then repeatedly either move to another configuration with those counters \(0\) without using \(Z_k\), or we can use \(Z_k\). The computation ends using \(\VAS_{k-1}\) and reaching the given target \(\vect{b}_t\).

Well-definedness:  We have to prove that in this expression \(\Expression_k\) the operations \(\circ, \cup, \ast\) are only used on matching dimensions. This follows since our specified targets and sources coincide. For example in the union \(\Expression_{k-1, 0^k, 0^k} \cup \Expression_{Z_k}\), both parts fix \(0^k, 0^k\) as required. %Observe that \(\binom{\VAS_{k-1}}{\vect{b}_s \times \vect{b}_t}\) is considered as subset of \(\N^{d'} \times \N^{d''}\) if \(\vect{b}_s \in \N^{d-d'}\) and \(\vect{b}_t \in \N^{d-d''}\). By the additional induction hypothesis, we therefore have for example \(in(\Expression_{k-1, \vect{b}_s, \vect{b}_t})=d'\), \(out(\Expression_{k-1, \vect{b}_s, \vect{b}_t})=d''\), \(in(\Expression_{k-1,0^k, 0^k})=d-k\) and \(out(\Expression_{k-1, 0^k, 0^k})=d-k\). As mentioned above, \(in(\Expression_{Z_k})=out(\Expression_{Z_k})=d-k\). It is easy to see that all the numbers fit, and the strengthened induction hypothesis is met. %We remark that if we had tried the same construction for a general Minsky machine, then this condition would not be met for \(k=2\), because \(\binom{\TransitionSet_1, f}{0^1, 0^1}\) would be monotone on the second coordinate, but \(L(Z_2)\) would be monotone on the first coordinate. Hence the union would not be monotone in any coordinate.

Correctness: It is clear that \(\Rel(\Expression_{k, \vect{b}_s, \vect{b}_t}) \subseteq \binom{\VAS_k}{\vect{b}_s, \vect{b}_t}\), since the expression describes a special form of runs from \(\vect{b}_s\) to \(\vect{b}_t\). For the other direction, let \((\vect{c}_0, \dots, \vect{c}_r)\) be a run of \(\VAS_k\) such that \(\vect{c}_0 \in \{\vect{b}_s\} \times \N^{d-d'}\) and \(\vect{c}_r \in \{\vect{b}_t\} \times \N^{d-d''}\). We have to show that \((\vect{c}_0, \vect{c}_r) \in \Rel(\Expression_{k, \vect{b}_s, \vect{b}_t})\). Case 1: The run does not use actions in \(Z_k\). Then the run shows membership in \(\binom{\VAS_{k-1}}{\vect{b}_s, \vect{b}_t} \subseteq \Rel(\Expression_{k-1, \vect{b}_s, \vect{b}_t}) \subseteq \Rel(\Expression_{k, \vect{b}_s, \vect{b}_t})\).

Case 2: The run does use \(Z_k\). Let \(\pi_k \colon \N^d \to \N^{d-k}\) be the projection to the last \(d-k\) coordinates, i.e. it removes the anyways fixed coordinates. Let \(i_1, \dots, i_s\) be the indices such that \(\vect{c}_{i_j} \to \vect{c}_{i_j+1}\) uses an action \(\vect{a}_j \in Z_k\), i.e. \((\pi_k(\vect{c}_{i_j}), \pi_k(\vect{c}_{i_j+1})) \in \Rel(\Expression_{Z_k})\). Then the part of the run \((\vect{c}_{i_j+1}, \dots, \vect{c}_{i_{j+1}})\) does not use any \(Z_k\) transitions. Hence \((\pi_k(\vect{c}_{i_j+1}), \pi_k(\vect{c}_{i_{j+1}})) \in \binom{\VAS_{k-1}}{0^k, 0^k} \subseteq \Rel(\Expression_{k-1, 0^k, 0^k})\) for all \(j \in \{1,\dots,s\}\). Hence we already obtain \((\pi_k(\vect{c}_{i_1}), \pi_k(\vect{c}_{i_s+1})) \in \Rel((\Expression_{k-1, 0^k, 0^k} \cup \Expression_{Z_k})^{\ast})\). Now similar to \(\pi_k\), let \(\pi_{d-d'}, \pi_{d-d''}\) be the projections removing the first \(d-d', d-d''\) coordinates. Since \((\vect{c}_0, \dots, \vect{c}_{i_1})\) does not use \(Z_k\) and \(\vect{c}_0 \in \{\vect{b}_s\} \times \N^{d-d'}\), we obtain \((\pi_{d-d'}(\vect{c}_0), \pi_k(\vect{c}_{i_1})) \in \binom{\VAS_{k-1}}{\vect{b}_s, 0^k} \subseteq \Rel(\Expression_{k-1, \vect{b}_s, 0^k})\) and similarly \((\pi_k(\vect{c}_{i_s+1}), \pi_{d-d''}(\vect{c}_r)) \in \binom{\VAS_{k-1}}{0^k, \vect{b}_t} \subseteq \Rel(\Expression_{k-1, 0^k, \vect{b}_t})\). Altogether we obtain \((\pi_{d-d'}(\vect{c}_0), \pi_{d-d''}(\vect{c}_r)) \in \Rel(\Expression_{k, \vect{b}_s, \vect{b}_t})\).

``\(\Leftarrow\)'': This follows by structural induction. The construction follows the standard conversion RegEx to \(\varepsilon\)-NFA, while adding some obvious zero tests.

As our definition of \(\Rel(\vectSet{Y})\) is representation independent, we choose to represent every \(\Rel(\vectSet{Y})\) using VASS of the same dimension \(d\), and with \(\vect{b}_s=0^{k}\) and \(\vect{b}_t=0^{l}\) for some \(k,l \in \N\). That every VASS section has a representation with \(\vect{b}_s=0^k\) and \(\vect{b}_t=0^l\) follows since one can simply add a new initial and final state \(q_{in}, q_{fin}\), and add \(\vect{b}_s\) when leaving \(q_{in}\) respectively subtract \(\vect{b}_t\) when entering \(q_{fin}\). To guarantee that all VASS have the same dimension \(d\), add unused counters. We will prove by induction that every subexpression \(\Expression'\) of the given starting expression \(\Expression\) has an equivalent PVASS section, where the PVASS has dimension \(d\) and \(\vect{b}_s=0^{d-in(\Expression)}, \vect{b}_t=0^{d-out(\Expression)}\). In the base case \(\Expression=\vectSet{Y}\) we have required this above.

\(\Expression \cup \Expression'\): By induction hypothesis, we can write \(\Rel(\Expression)=\binom{\VAS, p,q}{0^{d'},0^{d''}}\) and \(\Rel(\Expression')=\binom{\VAS', p',q'}{0^{d'},0^{d''}}\) using \(d\) dimensional PVASS \(\VAS, \VAS'\). That both sections use the same \(d'\) and \(d''\) follows by the restriction on expressions. Our new PVASS simply has a new initial and final state, and performs a non-deterministic choice whether to move to \(p\) and simulate \(\VAS\) or move to state \(p'\) and simulate \(\VAS'\).

Formally, write \(\VAS=(Q,E)\) and \(\VAS=(Q',E')\). We define the new PVASS \(\VAS''\) as \((Q \cup Q' \cup \{q_{in}, q_{out}\}, E \cup E' \cup \{(q_{in}, p), (q_{in}, p'), (q, q_{out}), (q', q_{out})\}\). All copied edges keep their labels, the four new edges get the label \(g(e)=0\), i.e. they do not zero test, and \(f(e)=0^d\), i.e. they do not change the counters. It is easy to see that \(\Rel(\Expression \cup \Expression')=\binom{\VAS'', q_{in}, q_{out}}{0^{d'}, 0^{d''}}\).

\(\Expression \circ \Expression'\): By induction hypothesis, we can write \(\Expression=\binom{\VAS, p,q}{0^{d-d'},0^{d-d_{mid}}}\) and \(\Expression'=\binom{\VAS', p',q'}{0^{d-d_{mid}},0^{d-d''}}\) using \(d\) dimensional PVASS \(\VAS, \VAS'\). Our new PVASS first simulates \(\VAS\), then checks that the first \(d_{mid}\) counters are \(0\), before simulating \(\VAS'\). Checking the intermediate configuration will be done using a zero test.

Formally, write \(\VAS=(Q,E)\) and \(\VAS'=(Q',E')\). We define the PVASS \(\VAS''\) as \((Q \cup Q', E \cup E' \cup \{(q,p')\}\), where prior edges keep their labels, and \(f(q,p')=0^d\) and \(g(q,p')=d-d_{mid}\), i.e. we zero test the first \(d-d_{mid}\) counters. It is easy to see that \(\Rel(\Expression \circ \Expression')=\binom{\VAS'',p,q'}{0^{d-d'}, 0^{d-d''}}\).

\(\Expression^{\ast}\): By induction hypothesis, we can write \(\Expression=\binom{\VAS, p,q}{0^{d-d'},0^{d-d'}}\), where both vectors have the same number of fixed coordinates by the restriction on the grammar. We now simply add a new initial state and an edge from \(q\) to this new initial state which performs zero tests on the first \(d-d'\) coordinates.

Formally, write \(\VAS=(Q,E)\) and define \(\VAS'=(Q \cup \{q_{in}\}, E \cup \{(q_{in}, p), (q,q_{in})\})\), where edges \(e\in E\) keep their labels, and \(f(q_{in}, p)=f(q,q_{in})=0^d\), i.e. counters are not changed, and \(g(q_{in},p)=g(q,q_{in})=d-d'\), i.e. the first \(d-d'\) counters are zero tested. It is easy to see that \(\Rel(\Expression^{\ast})=\binom{\VAS',q_{in}, p}{0^{d-d'},0^{d-d'}}\).
\end{proof}

The ``\(\Rightarrow\)'' direction breaks down for a Minsky machine. Namely for \(k=2\), one subexpression is \((\binom{\VAS_1}{0^2,0^2} \cup \Expression_{Z_2})^{\ast}\). The parts we take the union of do not have the same dimension, as the first coordinate should be free in \(\Expression_{Z_2}\) for a Minsky machine.

In future sections we will require expressions where \(\ast\) is only used on relations \(\vectSet{X}\) which are monotone. Surprisingly, we can without loss of generality require this. Before we prove this, let us first provide an example of a valid expression in the grammar where this fails.

\begin{example}
Let \(\VAS\) be the PVAS of dimension \(2\) without any transitions. Consider the expression \((\binom{\VAS}{\epsilon, 0^1} \circ \binom{\VAS}{0^1,\epsilon})^{\ast}\). The expression below the \(\ast\) says that you start with any configuration, fix the first counter to \(0\) and end with any configuration. Since the PVASS \(\VAS\) does not have any transitions, in order for the composition to be possible, you have to have already started with the first counter equal to \(0\), and also end with such a configuration.

Hence the better expression would be \((\binom{\VAS}{0^1, 0^1} \circ \binom{\VAS}{0^1,0^1})^{\ast}\). This expression fulfills the property that if \(\Expression'\) is a subexpression of \(\Expression^{\ast}\), then \(in(\Expression') \geq in(\Expression)\) and \(out(\Expression') \geq out(\Expression)\), i.e. interior nodes have fewer fixed coordinates. This will suffice for monotonicity.
\end{example}

\begin{restatable}{lemma}{LemmaMonotonicity}
Let \(d' \in \N\) and \(\Expression\) expression such that every subexpression \(\Expression'\) fulfills \(in(\Expression')\geq d'\) and \(out(\Expression') \geq d'\). Then \(\Rel(\Expression)\) is monotone in the \(j\)-th last coordinate for all \(j \leq d'\). \label{LemmaMonotonicity}
\end{restatable}

\begin{proof}
By structural induction. In the base case, \(\vectSet{Y}=\binom{\VAS, p,q}{\vect{b}_s, \vect{b}_t}\) for a VASS \(\VAS\). Since VAS reachability relations are monotone in every coordinate, when the last \(d'\) coordinates are neither fixed in the input nor output then \(\vectSet{Y}\) is monotone in these coordinates.

\(\Expression \cup \Expression'\): Relations monotone in the \(j\)-last coordinate are clearly stable under union.

\(\Expression \circ \Expression'\): Let \(j \leq d'\), \(d_1=in(\Expression), d_2=out(\Expression)\), \(d_3=out(\Expression')\). Since the subexpression \(\Expression\) fulfills \(in(\Expression)\geq d'\) and \(out(\Expression) \geq d'\) by assumption, \(\Rel(\Expression)\) is by induction monotone in the \(j\)-th last coordinate. Same for \(\Rel(\Expression')\). Let \((\vect{x}_1, \vect{x}_2) \in \Rel(\Expression)\) and \((\vect{x}_2, \vect{x}_3) \in \Rel(\Expression')\). By monotonicity in the \(j\)-th last coordinate we obtain \((\vect{x}_1+\vect{e}_{d_1+1-j}, \vect{x}_2+\vect{e}_{d_2+1-j})\in \Rel(\Expression)\) and \((\vect{x}_2+\vect{e}_{d_2+1-j}, \vect{x}_3 + \vect{e}_{d_3+1-j}) \in \Rel(\Expression')\). Hence \((\vect{x}_1+\vect{e}_{d_1+1-j}, \vect{x}_3 + \vect{e}_{d_3+1-j}) \in \Rel(\Expression \circ \Expression')\).

\(\Expression^{\ast}\): Let \(j \leq d'\). By repeatedly applying the case of \(\circ\), we obtain for all \(n \in \N\) that \(\Expression^n\) is monotone in the \(j\)-th last coordinate. Again by stability of such relations under union, the relation \(\Expression^{\ast}=\bigcup_{n\in \N} \Expression^n\) is monotone in the \(j\)-th last coordinate.
\end{proof}

\begin{corollary}
For every expression \(\Expression\), there is another expression \(\Expression'\) with \(\Rel(\Expression)=\Rel(\Expression')\), where \(\ast\) is only applied on monotone relations. \label{CorollaryStarOnMonotone}
\end{corollary}

\begin{proof}
Let \(\Expression\) be an expression. Apply first \(\Leftarrow\) and then \(\Rightarrow\) of Theorem \ref{TheoremEquivalenceRegExPVAS}. The constructed expression fulfills for every subexpression \(\Expression^{\ast}\) and subsubexpression \(\Expression'\) of \(\Expression^{\ast}\) that \(in(\Expression') \geq in(\Expression)\) and \(out(\Expression') \geq out(\Expression)\). Then by Lemma \ref{LemmaMonotonicity}, \(\Rel(\Expression)\) is monotone in every one of the last \(in(\Expression)=out(\Expression)\) coordinates, i.e. in every coordinate. Hence \(\Rel(\Expression)\) is monotone.
\end{proof}

\section{A Well-Quasi-Order on Runs of Priority VAS} \label{SectionWQO}

% !TeX root = Main.tex

Starting from this section, we also use \(\ast\) for repeated concatenation, not only transitive closure \(\ast\) of repeated composition. To distinguish them, we write \(\ast_{concat}\) for concatenation \(\ast\).

A partial order \((\vectSet{X}, \leq)\) is a reflexive, transitive and antisymmetric relation \(\leq \subseteq \vectSet{X} \times \vectSet{X}\). A set \(\vectSet{U} \subseteq \vectSet{X}\) is upward-closed if for all \(\vect{x}\in \vectSet{U}\) and all \(\vect{x}' \geq \vect{x}\) we have \(\vect{x}' \in \vectSet{U}\). Every subset \(\vectSet{X}' \subseteq \vectSet{X}\) is contained in a unique minimal upward-closed set \(\UpwardClosure{\vectSet{X}'}:=\{\vect{x}' \in \vectSet{X}\mid \exists \vect{x} \in \vectSet{X}' \colon \vect{x} \leq \vect{x}'\}\). A basis of an upward-closed set \(\vectSet{U}\) is a subset \(\vectSet{F} \subseteq \vectSet{U}\) such that \(\UpwardClosure{\vectSet{F}}=\vectSet{U}\).

A partial order is a well-quasi-order if every upward closed set \(\vectSet{U} \subseteq \vectSet{X}\) has a finite basis \(\vectSet{F}\). Or equivalently, for every infinite sequence \(\vect{x}_1, \vect{x}_2, \dots \subseteq \vectSet{X}\) there are indices \(i < j\) with \(\vect{x}_i \leq \vect{x}_j\), or equivalently there are indices \((i_m)_{m \in \N}\) such that \(\vect{x}_{i_m} \leq \vect{x}_{i_k}\) for all \(m \leq k\).

Most well-quasi-orders, in particular the ones we will need, are constructed from the following basic ordering by applying standard closure properties stated afterwards:

\begin{example}
Let \(\vectSet{F}\) be finite. Then the equality relation \(=\) is a well-quasi-order on \(\vectSet{F}\).
\end{example}

\begin{lemma}
Let \((\vectSet{X}_1, \leq_1), (\vectSet{X}_2, \leq_2)\) be wqo's. Then \(((\vectSet{X}_1 \cup \vectSet{X}_2, \leq_1 \cup \leq_2)\) is a wqo.\label{LemmaWQOUnion}
\end{lemma}

\begin{lemma}[Dickson's Lemma]
\((\vectSet{X}_1, \leq_1), (\vectSet{X}_2, \leq_2)\) wqo's \(\Rightarrow\) \((\vectSet{X}_1 \times \vectSet{X}_2, \leq_1 \times \leq_2)\) wqo. \label{LemmaDickson}
\end{lemma}

\begin{lemma}[Higman's Lemma]
Let \((\vectSet{\Sigma}, \leq)\) be a well-quasi-order. Then \((\Sigma^{\ast_{concat}}, \leq^{\ast_{concat}})\) is a well-quasi-order, where \(\leq^{\ast_{concat}}\) is the scattered subword ordering defined via \(w=(\vect{x}_1, \dots, \vect{x}_r) \leq w'=(\vect{y}_1, \dots, \vect{y}_s) \iff\) there exists an injective order preserving function \(f \colon \{1,\dots, r\} \to \{1,\dots, s\}\) such that \(\vect{x}_i \leq \vect{y}_{f(i)}\) for all \(i \in \{1,\dots, r\}\). \label{LemmaHigman}
\end{lemma}

\subsection{Well-Quasi-Order for VAS}

In order to explain the wqo for expressions, we start by defining Jancar's wqo ordering for runs of a VAS \(\VAS\). A run is no longer viewed as a sequence of configurations, but instead as element of \(\Omega(\VAS)=\N^{\dimD} \times (\N^{\dimD} \times E)^{\ast_{concat}} \times \N^{\dimD}\), where \(E\) is the set of edges of the VAS. The first part is \(\source(\rho)\), the last part is \(\target(\rho)\) and the middle parts are the steps of the run. One can extend this to states by replacing \(\N^{\dimD}\) by \(Q \times \N^{\dimD}\) everywhere and requiring states to coincide. \(\Omega(\VAS)\) is well-quasi-ordered by Lemma \ref{LemmaDickson} and Lemma \ref{LemmaHigman}. 

This wqo is carefully engineered to ensure that the relation \(\{\dirOfRun(\rho') \mid \rho' \geq \rho\}=:\dirOfRun(\rho)+\vectSet{P}_{\rho}\) has good properties. \(\vectSet{P}_{\rho}\) is called the \emph{transformer relation} of the run. A vector \((\vect{v}, \vect{w})\) is \emph{pumpable into \(\rho\)} if \((\vect{v},\vect{w}) \in \vectSet{P}_{\rho}\). The minimal property we want \(\vectSet{P}_{\rho}\) to fulfill is closure under addition, i.e. if a vector can be pumped once, then it can be pumped arbitrarily often. To understand why exactly the above expression ensures this, consider Figure \ref{FigureExampleVASRunOrdering}.

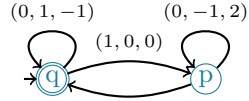
\begin{figure}[ht] 
	\centering % centers the figure
	\begin{tikzpicture}[-, auto, node distance=0.5cm]
		\tikzset{every place/.append style={minimum size=0.4cm, nicebluebright}}
		\tikzset{every label/.append style={text=nicebluebright}}
		\tikzset{every transition/.style={minimum size=0.25cm}}
		\tikzset{every edge/.append style={font=\scriptsize}}
		
		\newcommand*{\distancesubx}{2.5cm}
	
		% States
		\node[place, double] (A) at (0.5,0) {q};
		\node[place] (B) at (0.5+\distancesubx, 0) {p};
		\node[white!100] (C) at (0,0) {};
		
		% Edges
		\path[->, thick, looseness=8] (A) edge[] node[above] {\((0,1,-1)\)} (A);
		\path[->, thick, looseness=8] (B) edge[] node[above] {\((0,-1,2)\)} (B);
		\path[->, thick, out=20, in=160] (A) edge[] node[above] {\((1,0,0)\)} (B);
		\path[->, thick, out=200, in=-20] (B) edge[] node[below] {} (A);
		\path[->, thick] (C) edge[] (A);
		
	\end{tikzpicture}
	\caption{Typical example of a non-semilinear VAS \cite{HopcroftP79}. Edges \(e\) are only labelled with their update \(f(e)\), since this is a VASS, i.e. every edge fulfills \(g(e)=0\).}
	\label{FigureExampleVASRunOrdering}
\end{figure}

If we were to replace the middle part of the Jancar ordering with \((\N^{\dimD})^{\ast}\) instead, then the run \(q(0,1,1) \to p(1,1,1)\) would be smaller than the run \(q(0,1,1) \to q(0,2,0) \to p(1,2,0)\to p(1,1,2)\). Hence pumping properties would tell us that \(p(1,1,n)\) should be reachable for every \(n \geq 1\), which is obviously wrong. One might also consider the ordering without explicitly remembering the start and end configurations, but then some vectors with negative components might be claimed to be pumpable, since Higman might insert the smaller run in the middle. Pumping negative vectors is however trivially impossible.

Above we provided one way to define the transformer relation \(\vectSet{P}_{\rho}\). As shown in \cite[Lemma 7.5]{Leroux11}, there is an important equivalent characterization. Given a configuration \(\vect{c}\), one defines the transformer relation \(\vectSet{P}_{\vect{c}}\) for this configuration via \((\vect{x},\vect{y}) \in \vectSet{P}_{\vect{c}} \iff \vect{c}+\vect{x} \to_{\VAS}^{\ast} \vect{c}+\vect{y}\). Intuitively, you utilize the existence of configuration \(\vect{c}\) to \emph{transform} \(\vect{x}\) into \(\vect{y}\). Another intuition is that \(\vect{c}\) is a \emph{capacity} which allows us to slightly enter the negative (up to \(\vect{c}\)). This relation is a generalization of ``pumping possible via self-loop on a state''. The equivalent characterization of \(\vectSet{P}_{\rho}\) is: Write \(\rho=(\vect{c}_0, \dots, \vect{c}_r)\), then \(\vectSet{P}_{\rho}=\vectSet{P}_{\vect{c}_0} \circ \dots \circ \vectSet{P}_{\vect{c}_r}\). 

Consider the example in Figure \ref{FigureExampleVASS} in the introduction. At any configuration \(\vect{c}_1\) in state \(q_1\), we have \(\vectSet{P}_{\vect{c}_1}=\{(x,y,z),(x',y,z) \mid x' \geq x\}\), i.e. one can ``transform \(x\) into a larger number'' and at a configuration \(\vect{c}_2\) in state \(q_2\), we have \(\vectSet{P}_{\vect{c}_2}=\{(x,y,z),(x',y',z) \mid x+y=x'+y', x \geq x'\}\), i.e. one can ``transform any number of \(x\) into the same number of \(y\)''. This leads any run \(\rho\) through states \(q_s \to q_1 \to q_2 \to q_t\) to be able to arbitrarily increase \(x\) and \(y\), as mentioned in the introduction. Beware though that if a VASS has a complicated nested loop structure, then \(\vectSet{P}_{\vect{c}}\) can be non-semilinear.

\subsection{Well-Quasi-Order for Expressions}

We will now define for every expression \(\Expression\) a well-quasi-ordered set of runs \(\Omega(\Expression)\). We want the different segments of the runs to be labelled with which subexpression of \(\Expression\) they belong to. Hence let \(Tag\) be a set containing unique start and end labels \(\lambda_s(\Expression')\) and \(\lambda_t(\Expression')\) for every node \(\Expression'\) in the syntax tree of \(\Expression\). The set \(\Omega(\Expression)\) will be modelled after the Jancar ordering, though configurations might now have different dimensions. It is important to understand why \((\N^d \times E)^{\ast_{concat}}\) has to be used for the Jancar ordering to lead to nice properties. The answer: It is equivalent to using \((\N^d \times E \times \N^d)^{\ast_{concat}}\): When thinking of \(\to_{\VAS}^{\ast}\) as \(\{\to_{e_1} \cup \dots \cup \to_{e_m}\}^{\ast}\), this means every letter is supposed to be a run of the expression \(\Expression\) inside the \(\ast\). For a VASS, this means \((\source, \target) \in \to_{e}\) tagged with the choice of \(e\).

\begin{definition}
If \(\vectSet{Y} \subseteq \N^{d'} \times \N^{d''}\) is a VAS section represented by VAS \(\VAS\), \(\vect{b}_s, \vect{b}_t\), let \(\pi_{in}, \pi_{out}\) be the projections projecting away fixed coordinates of in- and output. Let \(\Omega_{\vect{b}_s,\vect{b}_t}\) be the set of runs \(\rho \in \Omega(\VAS)\), whose source and target have the correct values on fixed coordinates. We define \(\Omega(\vectSet{Y})=\{(\pi_{in}(\source(\rho)), \lambda_s(\vectSet{Y}) w \lambda_t(\vectSet{Y}), \pi_{out}(\target(\rho))) \mid \rho=(\source(\rho),w,\target(\rho)) \in \Omega_{\vect{b}_s, \vect{b}_t}(\VAS)\}\). I.e. we consider runs of \(\VAS\) with source and target adhering to the fixed coordinates, add markers in the word w, and instead of storing the full source and target, we only store the non-fixed coordinates. For \(\rho \in \Omega(\vectSet{Y})\), we refer to the projected configurations as \(\source(\rho)\) and \(\target(\rho)\).

We define \(\Omega(\Expression_1 \cup \Expression_2)=\Omega(\Expression_1) \cup \Omega(\Expression_2)\), i.e. we simply take unions of the sets of runs.

We define \(\Omega(\Expression_1 \circ \Expression_2)=\{\lambda_s(\Expression_1) \rho_1 \lambda_t(\Expression_1) \lambda_s(\Expression_2) \rho_2 \lambda_t(\Expression_2) \mid \target(\rho_1)=\source(\rho_2)\} \subseteq Tag \times \Omega(\Expression_1) \times Tag^2 \times \Omega(\Expression_2) \times Tag\). I.e. we concatenate the runs if possible and use markers.

We define \(\Omega(\Expression^{\ast})=\{\lambda_s(\Expression)\rho_1\lambda_t(\Expression) \dots \lambda_s(\Expression) \rho_n \lambda_t(\Expression) \mid n \in \N, \rho_1, \dots, \rho_n \in \Omega(\Expression), \target(\rho_i)=\source(\rho_{i+1})\} \subseteq \N^{in(\Expression^{\ast})} \times (\Omega(\Expression) \cup Tag)^{\ast_{concat}} \times \N^{out(\Expression^{\ast})}\). I.e. we consider all concatenations of any length \(n\in \N\) and add tags splitting the different parts \(\rho_i\).
\end{definition}
\begin{definition}
We define a wqo \(\leq_{\Omega(\Expression)}\) on \(\Omega(\Expression)\) recursively. For VAS sections \(\vectSet{Y}\) we use the Jancar ordering, observing that fixed coordinates coincide for every run, and can hence be ignored. For the recursive definition we use Lemmas \ref{LemmaWQOUnion}, \ref{LemmaDickson} and \ref{LemmaHigman}. 
\end{definition}

Whenever we concatenate runs, we do not write the tags, because they can be inferred. Their existence is however important, as we can see for example for \(\Expression_1 \circ \Expression_2\): We have \(\rho \leq_{\Omega(\Expression_1 \circ \Expression_2)} \rho'\) if and only if \(\rho=\rho_1 \rho_2, \rho'=\rho_1' \rho_2'\) with \(\rho_1, \rho_1' \in \Omega(\Expression_1), \rho_2, \rho_2' \in \Omega(\Expression_2)\) such that \(\rho_1 \leq_{\Omega(\Expression_1)} \rho_1'\) and \(\rho_2 \leq_{\Omega(\Expression_2)} \rho_2'\). This is where the tags will become important: From the tags, we can infer how \(\rho\) is supposed to be split into \(\rho_1\) and \(\rho_2\), and similarly for \(\rho'\). Let us give a different example for why we need the tags for every subexpression. Imagine we consider the expression \(\vectSet{Y} \cup \vectSet{Y}'\), where \(\Rel(\vectSet{Y})=\N \cdot (2,1)\) and \(\Rel(\vectSet{Y}')=\N \cdot (1,2)\). If we wrote the empty run \(\rho\) as \((0, \epsilon, 0)\), i.e. did not label it, it would not be clear whether it can pump \(\N \cdot (2,1)\) or \(\N \cdot (1,2)\), this depends on which subexpression it belongs to.

\subsection{Comparison with other Well-Quasi-Orders}

In prior literature, some wqos for Priority VAS \cite{Bonnet11, Bonnet12} and even for the more general model of Grammar VAS \cite{LerouxPSS19} were introduced. In this subsection we compare our wqo to theirs.

In a grammar VAS paths are restricted to a given context-free grammar \(G\), and it is known that Priority VAS correspond to the subclass of Grammar VAS, where the grammar is thin/finite index. In a thin grammar, every non-terminal has a rank (called index) which can only decrease and for every production \(X \to YZ\), either \(\rank(Y)<\rank(X)\) or \(\rank(Z)<\rank(X)\), i.e. only one of the produced non-terminals can have the same rank.

The equivalence of PVASS with these grammars can in fact be seen via our RegEx characterization: One can implement \(\Expression^{\ast}\) via \(S \to X S\), where \(X\) implements \(\Expression\). Similarly composition \(\Expression \circ \Expression'\) can be implemented via \(S \to X X'\) where \(X\) implements \(\Expression\), and \(X'\) implements \(\Expression'\). In \cite{LerouxPSS19} a wqo on runs based on Kruskal’s tree ordering on syntax trees is defined. It can be shown that their ordering for the grammar obtained from a RegEx coincides with our ordering. In fact this is another motivation for the markers we use for splitting runs: Syntax trees naturally distinguish between being in the left or right branch of the RegEx. In \cite{LerouxPSS19} however they did not manage to show some of the pumping properties of the well-quasi-order which we require, and will be able to prove using our RegEx.

Another example are Bonnet's works \cite{Bonnet11, Bonnet12}. His well-quasi-order is based on a repeated application of Higman's Lemma, similar to our ordering. The only difference is that in \cite{Bonnet12} the finest split of any run is chosen. In our terminology, if the expression is \(\Expression^{\ast}\), where \(\Expression\) requires the first coordinate to be \(0\), then a run \(\rho=\rho_1 \rho_2\) is always split as \(\rho_1, \rho_2 \in \Omega(\Expression)\) in \cite{Bonnet12}, i.e. \(\rho \in \Omega(\Expression^2)\). While in our case also \(\rho \in \Omega(\Expression)\) is possible, this is determined by the markers. When limitting ourselves to runs with the finest split the orderings coincide.

% !TeX root = Main.tex

\subsection{Geometric Preliminaries} \label{SectionGeometricPreliminaries}

In this section we repeat some definitions from the VASS literature, pertaining to the pumping properties the relations \(\vectSet{P}_{\rho}\) fulfill. Readers familiar with the notions can skip this section. For a visual representation of the geometric definitions see e.g. \cite{GuttenbergRE23}. We state definitions for sets, they apply to relations \(\vectSet{R} \subseteq \Q^{d'} \times \Q^{d''}\) by viewing them as set \(\vectSet{R} \subseteq \Q^{d'+d''}\).

\subparagraph{Cones and periodic sets.}
A set \(\vect{C} \subseteq \mathbb{Q}^{\dimD}\) is a \emph{cone} if \(\vect{0}\in \vect{C}\), \(\vect{C}+\vect{C} \subseteq \vect{C}\) and \(\Q_{>0}\vect{C} \subseteq \vect{C}\). Given a set \(\vect{F} \subseteq \mathbb{Q}^{\dimD}\), the cone generated by \(\vect{F}\) is the smallest cone containing \(\vect{F}\). 

A cone \(\vect{C}\) is \emph{definable} if it is definable in \(\FO(\mathbb{Q}, +, \geq)\).

A set \(\vect{P} \subseteq \mathbb{N}^{\dimD}\) is a \emph{periodic set} if \(\vect{P} + \vect{P} \subseteq \vect{P}\) and \(\vect{0} \in \vect{P}\). For any set \(\vect{F} \subseteq \N^{\dimD}\), the periodic set \(\vect{F}^{\ast}\) generated by \(\vect{F}\) is the smallest periodic set containing \(\vect{F}\). %Explicitly we have \(\vect{F}^{\ast}=\{\vect{p}_1+\dots+\vect{p}_r \mid r\in \N, \vect{p}_i \in \vect{F} \text{ for all } i\}\). 
A periodic set \(\vect{P}\) is \emph{finitely generated} if \(\vect{P}=\vect{F}^{\ast}\) for some finite set \(\vect{F}\).%Finitely generated periodic sets are characterized as follows:

Finitely generated periodic sets provide an equivalent way to define linear sets as sets of the form \(\vect{b}+ \vectSet{P}\), where \(\vect{b} \in \N^{\dimD}\) and \(\vectSet{P} \subseteq \N^d\) is a finitely generated periodic set.

\subparagraph{Smooth Periodic Sets.}

The periodic relation \(\vectSet{P}_{\rho}\) for a run \(\rho\) of a VASS is rarely finitely generated, but it is smooth, a class 
introduced by Leroux in \cite{Leroux13}. In order to define smooth, we first reintroduce the set of directions of a periodic set. 

\begin{definition}{\cite{Leroux13, GuttenbergRE23}}
\label{def:directions}
Let \( \vect{P} \) be a periodic set.  A vector \(\vect{v} \in \Q^{\dimD}\) is a \emph{direction} of  \(\vect{P}\) if there exists \(m \in \N_{>0}\) and a point \(\vect{x}\) such that \(\vect{x}+\N \cdot m \vect{v} \subseteq \vect{P}\), i.e. some line in direction \(\vect{v}\) is fully contained in \(\vect{P}\). The set of directions of \(\vect{P}\) is denoted \(\dir(\vect{P})\).
\end{definition}

We can now define smooth periodic sets.

\begin{definition}{\cite{Leroux13, GuttenbergRE23}}
\label{def:smooth}
Let \( \vect{P} \) be a periodic set.
\begin{itemize}
\item \(\vect{P}\) is \emph{asymptotically definable} if \(\dir(\vect{P})\) is a definable cone.

\item \(\vect{P}\) is \emph{well-directed} if every sequence \((\vect{p}_m)_{m \in \N}\) of vectors \(\vect{p}_m \in \vect{P}\) has an infinite subsequence \((\vect{p}_{m_k})_{k \in \N}\) such that \(\vect{p}_{m_j}+\N(\vect{p}_{m_k} - \vect{p}_{m_j}) \subseteq \vectSet{P}\) for all \(k \geq j\). 

\item \(\vect{P}\) is \emph{smooth} if it is asymptotically definable and well-directed.
\end{itemize}
\end{definition}

\begin{example}
\label{ex:smooth}
Examples of smooth periodic sets are \(\vectSet{P}_1=\{(0,0)\}  \cup (1,1)+\N^2\) and \(\vectSet{P}_2=\{(x,y) \in \N^2 \mid y \leq x^2\}\). We have \(\dir(\vectSet{P}_2) \setminus \{(0,0)\}=\{(x,y)\in \Q_{\geq 0}\mid x >0 \}\). I.e. except pure north, every vector in \(\Q_{\geq 0}^2\) is a direction of \(\vectSet{P}_2\). On the other hand \(\dir(\vectSet{P}_1)=\Q_{\geq 0}^2\). 

\(\vectSet{P}_2\) is a very typical example: One idea with \(\dir(\vectSet{P})\) is to store the asymptotic steepness of the upper function, and ignore whether it is exponential or quadratic if it is superlinear.
\end{example}

\begin{example}
Examples of non-smooth sets are \(\vect{P}_1'=\{(x,y) \mid x \geq \sqrt{2} y\}\)  and \(\vect{P}_2'=(\{(0,1)\} \cup \{(2^m,1) \mid m \in \N\})^{\ast}=\{(x,n) \in \N^2 \mid x \text{ has at most } n \text{ bits set to }1\text{ in the binary representation.}\}\). $\vect{P}_1'$  is not asymptotically definable, because defining \(\dir(\vect{P}_1')\) requires irrationals, while $\vect{P}_2'$ is not well-directed (see observation 2 below). 
\end{example}

%Intuitively, the ``boundaries'' of a smooth periodic set in two dimensions are either straight lines or function graphs ``curving outward'',  as in the example on the right of Figure \ref{FigureExamplesNonFinitelyGeneratedPeriodicSet}. 

We make a few observations:
\begin{enumerate}
\item The set \(\dir(\vect{P})\) is a cone. It is by definition closed under non-negative scalar multiplication (due to the \(m\) in the definition). Furthermore, if two lines in different directions \(\vect{v}\) and \(\vect{v}'\) are contained in \(\vect{P}\), then by periodicity \(\vect{P}\) also contains a \(\vect{v},\vect{v}'\) plane, and so \(\vect{P}\) contains a line in every direction between \(\vect{v}\) and \(\vect{v}'\). For more details see \cite[Lemma V.7]{Leroux13}. \(\dir(\vectSet{P})\) should be viewed as a kind of ``limit cone'' containing \(\vectSet{P}\), it is however only one possible definition for a ``limit cone'' of \(\vectSet{P}\), other cones were considered in prior papers \cite{Leroux11, Leroux12}.
%\item Since the closed cone \(\overline{\mathbb{Q}_{\geq 0}P}$ is equal to $\overline{\dir(P)}\) (proven in Lemma \ref{LemmaDirectionsContainInterior}), if $P$ is smooth then $\overline{\mathbb{Q}_{\geq 0}P}$ is definable and therefore finitely generated. 
\item The definition of well-directed is stated this way to relate to wqo's, but the most important case of definition \ref{def:smooth} is when the \(\vect{p}_m\) are all on the same infinite line \(\vect{x}+\vect{v} \cdot \N\). Then the definition equivalently states that \(\vect{v} \in \dir(\vect{P})\), i.e. some infinite line in direction \(\vect{v}\) is contained in \(\vect{P}\). This makes sets where points are ``too scarce'' non-smooth. For instance, the set $\vect{P}_2$ of Example \ref{ex:smooth} contains infinitely many points on a horizontal line, but no full horizontal line, which would correspond to an arithmetic progression.
\end{enumerate}

%It is important to understand how the set of directions \(\dir(\vectSet{P})\) compares to the cone \(\Q_{\geq 0}\vectSet{P}\) generated by \(\vectSet{P}\).
%
%\begin{lemma}{\cite[Lemma 3.4]{GuttenbergRE23}}
%Let \(\vect{P}\) periodic. Then \(\interior(\overline{\mathbb{Q}_{\geq 0}\vect{P}}) \subseteq \mathbb{Q}_{\geq 0}\vect{P} \subseteq \dir(\vect{P}) \subseteq \overline{\mathbb{Q}_{\geq 0}\vect{P}}.\)
%
%In particular, all these sets have the same closure. \label{LemmaDirectionsContainInterior}
%\end{lemma}
%
%\begin{example}
%The set on the left of Figure \ref{FigureExamplesNonFinitelyGeneratedPeriodicSet} satisfies $\interior(\overline{\mathbb{Q}_{\geq 0}\vect{P}}) = \mathbb{Q}_{\geq 0}\vect{P} \subsetneq \dir(\vect{P}) = \overline{\mathbb{Q}_{\geq 0}\vect{P}}$. 
%Indeed, $\interior(\overline{\mathbb{Q}_{\geq 0}\vect{P}})$ contains every direction except north and east, but they both belong to \(\dir(\vect{P})\). 
%The middle set satisfies $\interior(\overline{\mathbb{Q}_{\geq 0}\vect{P}}) \subsetneq \mathbb{Q}_{\geq 0}\vect{P} = \dir(\vect{P}) \subsetneq \overline{\mathbb{Q}_{\geq 0}\vect{P}}$, since $\interior(\overline{\mathbb{Q}_{\geq 0}\vect{P}})$ contains neither north nor east, $\dir(\vect{P})$ contains east, and $\overline{\mathbb{Q}_{\geq 0}\vect{P}}$ contains both.
%\end{example}
%
%As mentioned before for smooth periodic sets \(\vectSet{P}\), \(\Fill(\vectSet{P})\) is finitely generated.
%
%\begin{proposition}{\cite[Lemma 5.1]{Leroux11}}
%\label{prop:Fillisfinitelygenerated}
%Let \(\vect{P}\) be smooth. Then \(\Fill(\vect{P})\) is finitely generated.
%\end{proposition}

\subparagraph{Almost semilinear relations.} \label{SubsectionAlmostSemilinear}

We reintroduce almost-semilinear sets, introduced by Leroux in \cite{Leroux11, Leroux12, Leroux13}. Intuitively, they generalize semilinear sets by replacing finitely generated periodic sets with smooth periodic sets.

\begin{definition}{\cite{Leroux12,Leroux13}}
A set \(\vect{X}\) is \emph{almost linear} if \(\vect{X}=\vect{b}+\vect{P}\), where \(\vect{b}\in \mathbb{N}^d\) and \(\vect{P}\) is a smooth periodic set, and \emph{almost semilinear} if it is a finite union of almost linear sets.
\end{definition}

It was shown in \cite{Leroux12,Leroux13} that VAS reachability sets/relations are almost semilinear. However,
it is easy to find almost semilinear sets that are not reachability sets of any VAS. One reason is that the definition of a smooth periodic set only restricts the ``asymptotic behavior'' of the set, which can be ``simple'' even if the set itself is very ``complex''.  

\begin{example}
Let \(\vect{X} \subseteq \N_{>0}\) be any set, for example \(\vect{X}:=\{m \in \N \mid m \text{ is Gödel-number of non-halting TM}\}\). Then \(\vect{P}:=\{(0,0)\} \cup (\{1\} \times \vect{X}) \cup \N_{>1}^2\) is a smooth periodic set. Indeed, it contains a line in every direction, and is thus well-directed and asymptotically definable. \label{ExampleUglyPeriodicSet}
\end{example}

To eliminate these types of sets Leroux required that every intersection of the set with a semilinear set is still almost semilinear. For instance, the intersection of the set $\vect{X}$ in Example \ref{ExampleUglyPeriodicSet} and the linear set \((1,0)+(0,1) \cdot \N\) is not almost semilinear. This leads to the following main theorems of \cite{Leroux13}, which we want to extend to Priority VAS:

\begin{theorem}{\cite[Theorem IX.1]{Leroux13}}
For every semilinear relation \(\vectSet{S}\) and reachability relation \(\vectSet{R}\) of a VAS, \(\vect{R} \cap \vect{S}\) is a finite union of relations \(\vect{b}+\vectSet{P}\), where \(\vectSet{P}\) is smooth periodic and for every linear relation \(\vectSet{L} \subseteq \vect{b}+\vectSet{P}\) there exists a \(\vect{p} \in \vect{P}\) such that \(\vect{p}+\vectSet{L}\) is flattable. \label{TheoremVASPetri}
\end{theorem}

Since projecting away fixed coordinates preserves almost semilinearity, namely the periodic sets are anyways \(0\) on fixed coordinates, this theorem also holds for VAS sections. 

As mentioned in the introduction, the same paper proceeded to prove the following:

\begin{theorem}{\cite[Theorem XI.2]{Leroux13}}
The reachability relation of a VAS is flattable if and only if it is semilinear. \label{TheoremFlattableVAS}
\end{theorem}

The hard part is of course to prove that semilinear implies flattable. Let us quickly recap how to obtain Theorem \ref{TheoremFlattableVAS} from Theorem \ref{TheoremVASPetri} as our main proof will not repeat this step, we will stop at obtaining Theorem \ref{TheoremVASPetri} for Priority VAS.

Leroux in \cite{Leroux13} defines a dimension \(\dim(\vectSet{S})\in \N\) of semilinear sets \(\vectSet{S}\). The important aspect of the dimension is that for every linear set \(\vectSet{L}\), we have \(\dim(\vectSet{L} \setminus (\vect{p}+\vectSet{L})) < \dim(\vectSet{L})\). For example \(\N^2 \setminus [(x,y)+\N^2]\) for any \(x,y \in \N\) is a finite union of lines, and hence 1-dimensional. The proof then proceeds by induction on the dimension of \(\vectSet{X}:=\vectSet{R} \cap \vectSet{S}\). 

Base \(k=0\): Since \(0\)-dimensional implies finite, such sets are flattable.

Step \(k \to k+1\): Let \(\vectSet{X}=\vectSet{R} \cap \vectSet{S}\) semilinear with \(\dim(\vectSet{X})=k+1\). We have to show that \(\vectSet{X}\) is flattable. Since flattable relations are closed under union, we can assume that \(\vectSet{X}\) is not only semilinear but even linear. Using \(\vectSet{X}\) as \(\vectSet{L}\) in Theorem \ref{TheoremVASPetri} (this requires combining almost linear components correctly, and hence some fiddling) there exists a vector \(\vect{p}\) such that \(\vect{p}+\vectSet{X}\) is flattable. Since \(\dim(\vectSet{X} \setminus (\vect{p}+\vectSet{X}))<\dim(\vectSet{X})\), \(\vectSet{X} \setminus (\vect{p}+\vectSet{X})\) is flattable by induction. Since flattable relations are closed under union, \(\vectSet{X}=[\vectSet{X} \setminus (\vect{p}+\vectSet{X})] \cup [\vect{p}+\vectSet{X}]\) is flattable.

With a similar induction on the dimension, Leroux obtained the following:

\begin{theorem}{\cite[Theorem 9.2]{Leroux11}}
Let \(\vectSet{R}\) be reflexive, transitive and such that for every semilinear \(\vectSet{S}\), \(\vectSet{R} \cap \vectSet{S}\) is almost semilinear. Then \(\vectSet{R}\) admits semilinear inductive invariants. \label{TheoremInductiveInvariants}
\end{theorem}

As a corollary of Theorem \ref{TheoremInductiveInvariants} and the PVAS version of Theorem \ref{TheoremVASPetri}, namely Theorem \ref{TheoremPVASPetri} of the next section, we obtain the following.

\begin{corollary}
Let \(\VAS\) be a PVASS, and \(\vect{C}_s, \vect{C}_t\) two configurations such that \(\vect{C}_s \not \to^{\ast} \vect{C}_t\). Then there exists a semilinear inductive invariant \(\vectSet{S}\) such that \(\vect{C}_s \in \vectSet{S}, \vect{C}_t \not \in \vectSet{S}\).
\end{corollary}

\begin{proof}
Let \(\vectSet{R}\) be the reachability relation of \(\VAS\). Clearly \(\vectSet{R}\) is reflexive and transitive. By Theorem \ref{TheoremPVASPetri}, every intersection \(\vectSet{R} \cap \vectSet{S}\) with a semilinear \(\vectSet{S}\) is a finite union of \(\vect{b}+ \vectSet{P}\), where \(\vectSet{P}\) is smooth periodic, i.e. \(\vectSet{R} \cap \vectSet{S}\) is almost semilinear. Hence \(\vectSet{R}\) fulfills all assumptions of Theorem \ref{TheoremInductiveInvariants}, whose conclusion is the existence of a separating inductive invariant \(\vectSet{S}\) as claimed.
\end{proof}

\section{Semilinear Priority VAS are Flattable} \label{SectionProofFinalTheorem}

% !TeX root = Main.tex

In this section we define flattability and prove that Theorem \ref{TheoremVASPetri} holds also for PVAS. We work on expressions \(\Expression\), and hence define flattability for expressions \(\Expression\) by structural induction, instead of defining flattability on PVAS directly.
\begin{definition}
Base case: \(\vectSet{R} \subseteq \Rel(\vectSet{Y})\) is flattable w.r.t. \(\vectSet{Y}\) if and only if \(\vectSet{R}\) is semilinear.

For \(\Expression_1 \cup \Expression_2\), a relation \(\vectSet{R} \subseteq \Rel(\Expression_1 \cup \Expression_2)\) is flattable w.r.t. \(\Expression_1 \cup \Expression_2\) if and only if there exist relations \(\vectSet{R}_i \subseteq \Rel(\Expression_i)\) flattable w.r.t. \(\Expression_i\) such that \(\vectSet{R} \subseteq \vectSet{R}_1 \cup \vectSet{R}_2\).

For \(\Expression_1 \circ \Expression_2\), a relation \(\vectSet{R} \subseteq \Rel(\Expression_1 \circ \Expression_2)\) is flattable w.r.t. \(\Expression_1 \circ \Expression_2\) if and only if there exist relations \(\vectSet{R}_i \subseteq \Rel(\Expression_i)\) flattable w.r.t. \(\Expression_i\) such that \(\vectSet{R} \subseteq \vectSet{R}_1 \circ \vectSet{R}_2\).

For \(\Expression^{\ast}\), remember that by Corollary \ref{CorollaryStarOnMonotone} we can assume that \(\Rel(\Expression)\) is monotone. Let \(in(\Expression)=d'\). We first make a preliminary definition: Given a vector \(\vect{v}=(\source,\target) \in \Rel(\Expression)\), its closure under monotonicity is \(m(\vect{v})=\{\vect{v}\}+\N (\vect{e}_1, \vect{e}_1)+\dots +\N (\vect{e}_{d'}, \vect{e}_{d'})\), where \(\vect{e}_{i} \in \N^{d'}\) is the \(i\)-th unit vector. We define the monotone transitive closure of \(\vect{v}\) as \(\mtc(\vect{v})=m(\vect{v})^{\ast}\), i.e. it is the relation of source and target configurations, such that the target can be reached by repeatedly applying only \(\vect{v}\), potentially at a larger configuration. A linear path scheme \(\vectSet{S}\) is a relation which can be written as \(\mtc(\source_1, \target_1) \circ \dots \circ \mtc(\source_r, \target_r)\) with \((\source_i, \target_i) \in \Rel(\Expression)\). This has to be defined using \(\mtc\) since for expressions we do not have a finite set \(E\) of ``edges'' anymore. But we want to express that the same edge sequences are taken.

A relation \(\vectSet{R} \subseteq \Rel(\Expression^{\ast})\) is flattable w.r.t. \(\Expression^{\ast}\) if and only if there exist finitely many relations \(\vectSet{R}_1, \dots, \vectSet{R}_k \subseteq \Rel(\Expression)\) flattable w.r.t. \(\Expression\) and linear path schemes \(\vectSet{S}_0, \dots, \vectSet{S}_k\) such that \(\vectSet{R} \subseteq \vectSet{S}_0 \circ \vectSet{R}_1 \circ \vectSet{S}_1 \circ \vectSet{R}_2 \circ \dots \circ \vectSet{R}_k \circ \vectSet{S}_k\). I.e., in terms of the ideas in the introduction, we have \(k\) loops where we are allowed to adapt the ``inner'' transitions, those are reflected by \(\vectSet{R}_i\). In between those loops, we are allowed to use linear path schemes of the outer loop.
\end{definition}

Let us mention some basic properties of this definition, proved in Appendix \ref{SectionAppendixProofFinalTheorem}. In particular, we compare this definition with the transition word definition of flattable.

\begin{definition}
Let \(\VAS=(Q,E)\) be a PVAS. For \(w=(e_1, \dots, e_k) \in E^{\ast}\) define 

\(\to_w:=\to_{e_1} \circ \dots \circ \to_{e_k}\). Let \(\to_w^{\ast}\) denote the reflexive and transitive closure of \(\to_w\).

A relation \(\vectSet{R}\) is transition word flattable w.r.t. \(\VAS\) if there exist transition sequences \(w_1, \dots, w_r \in E^{\ast}\) such that \(\vectSet{R} \subseteq \to_{w_1}^{\ast} \circ \dots \circ \to_{w_r}^{\ast}\).
\end{definition}

\begin{restatable}{lemma}{LemmaBasicExpressionflattability}
Let \(\Expression\) be an expression.

\begin{enumerate}
\item If \(\rho' \geq_{\Omega(\Expression)} \rho\) are runs, then \(\dir(\rho)+\N(\dir(\rho')-\dir(\rho))\) is flattable w.r.t. \(\Expression\), and all the corresponding runs are \(\geq_{\Omega(\Expression)} \rho\).
\item If \(\vectSet{R}, \vectSet{R}' \subseteq \Rel(\Expression)\) are flattable w.r.t. \(\Expression\), then also \(\vectSet{R} \cup \vectSet{R}'\) is.
\item Let \(\VAS\) be a PVAS such that \(\Expression\) is its expression in Theorem \ref{TheoremEquivalenceRegExPVAS}. If a relation \(\vectSet{R} \subseteq \Rel(\Expression)\) is flattable in our sense, then \(\vectSet{R} \subseteq \CanReach{\VAS}\) is transition word flattable.
\end{enumerate} \label{LemmaBasicflattabilityProperties}
\end{restatable}

In our proofs we will need to distribute \(+\) over \(\circ\).

\begin{restatable}{lemma}{LemmaExchangePlusAndComposition}
Let \((\vectSet{R}_i)_{i=1}^r \subseteq \N^{d'} \times \N^{d_{mid}}\) and \((\vectSet{R}_i')_{i=1}^r \subseteq \N^{d_{mid}} \times \N^{d''}\) be relations. Then \(\sum_{i=1}^r (\vectSet{R}_i \circ \vectSet{R}_i') \subseteq (\sum_{i=1}^r \vectSet{R}_i) \circ (\sum_{i=1}^r \vectSet{R}_i')\). \label{LemmaExchangePlusAndComposition}
\end{restatable}

\subsection{Proof Outline for Theorem \ref{TheoremPVASPetri}} \label{SubSectionFlatteningLinesSuffices}

In this subsection we provide a proof outline for Theorem \ref{TheoremPVASPetri}. As remarked in Section \ref{SubsectionAlmostSemilinear}, Theorem \ref{TheoremPVASPetri} suffices to obtain semilinear inductive invariants and flattability.

\begin{theorem}
For every semilinear relation \(\vectSet{S}\) and reachability relation \(\vectSet{R}\) of a PVAS, \(\vect{R} \cap \vect{S}\) is a finite union of relations \(\vect{b}+\vectSet{P}\), where \(\vectSet{P}\) is smooth periodic and for every linear relation \(\vectSet{L} \subseteq \vect{b}+\vectSet{P}\) there exists a \(\vect{p} \in \vect{P}\) such that \(\vect{p}+\vectSet{L}\) is flattable. \label{TheoremPVASPetri}
\end{theorem}

The starting point is to extend the definition of the transformer relation to expressions. Let \(\Expression\) be an expression, and \(\rho \in \Omega(\Expression)\) a run. Then we define the \emph{transformer relation of \(\rho\) w.r.t. \(\Expression\)} via \(\vectSet{P}_{\Expression,\rho}:=\{\dirOfRun(\rho')-\dirOfRun(\rho) \mid \rho' \geq \rho\}\), exactly as in the VAS case. The corresponding almost semilinear component is \(\Component(\Expression,\rho)=\dirOfRun(\rho)+\vectSet{P}_{\Expression,\rho}\).

The outline now consists of two parts: First we reduce Theorem \ref{TheoremPVASPetri} to Theorem \ref{TheoremFlatteningLines} and Lemma \ref{LemmaIntersectionSemilinear}. The closure property of Lemma \ref{LemmaIntersectionSemilinear} is easy to see, we give a proof in Appendix \ref{SectionAppendixProofFinalTheorem}. Hence afterwards the outline will focus on proving Theorem \ref{TheoremFlatteningLines}.

\begin{theorem}
Let \(\Expression\) be an expression, \(\rho \in \Omega(\Expression)\). Then \(\vectSet{P}_{\Expression,\rho}\) is smooth, periodic and:

\begin{enumerate}
\item Let \(\vectSet{R}_1, \vectSet{R}_2 \subseteq \Component(\Expression,\rho)\) flattable. Then \(\vectSet{R}_1+\vectSet{R}_2 - \dirOfRun(\rho)\) is flattable.
\item Every direction of \(\Component(\Expression,\rho)\) is flattable, i.e. for every \((\vect{e}, \vect{f}) \in \dir(\vectSet{P}_{\Expression, \rho})\) there exists \((\vect{a},\vect{b}) \in \vectSet{P}_{\Expression,\rho}\), \(n \in \N\) such that \(\dirOfRun(\rho)+(\vect{a}, \vect{b})+\N n(\vect{e}, \vect{f})\) is flattable w.r.t. \(\Expression\).
\end{enumerate} \label{TheoremFlatteningLines}
\end{theorem}

Property 2. is rather self-explanatory, important is property 1. The statement of property 1. is that if two relations are flattable using the same minimal run \(\rho\), then not just the sum is flattable, but even the sum \emph{ignoring the base point} \(\dirOfRun(\rho)\) is flattable.

\begin{restatable}{lemma}{LemmaIntersectionSemilinear}
Let \(\vectSet{S} \subseteq \N^{d'} \times \N^{d''}\) be semilinear, and \(\vectSet{X} \subseteq \N^{d'} \times \N^{d''}\) a PVAS section. Then \(\vectSet{X} \cap \vectSet{S}\) is a PVAS section, and hence has an equivalent expression \(\Rel(\Expression)\) by Theorem \ref{TheoremEquivalenceRegExPVAS}. \label{LemmaIntersectionSemilinear}
\end{restatable}

\begin{proof}[Proof of Theorem \ref{TheoremPVASPetri}]
Use Lemma \ref{LemmaIntersectionSemilinear} to obtain an expression \(\Expression\) with \(\Rel(\Expression)=\vectSet{R} \cap \vectSet{S}\). We write \(\Rel(\Expression)=\bigcup_{\rho \in \Omega(\Expression)_{min}} \dirOfRun(\rho)+\vectSet{P}_{\Expression, \rho}\) as mentioned in the introduction. By Theorem \ref{TheoremFlatteningLines}, these periodic relations are smooth, hence only the flattability claim is left. Let \(\vectSet{L} \subseteq \dirOfRun(\rho)+\vectSet{P}_{\Expression, \rho}\) linear, i.e. \(\vectSet{L}=\vect{b}+\N \vect{p}_1+\dots +\N \vect{p}_r\). By property 2. in Theorem \ref{TheoremFlatteningLines}, there exist \(n_i \in \N\) and \((\vect{a}_i, \vect{b}_i) \in \vectSet{P}_{\Expression, \rho}\) such that \(\vectSet{R}_i:=\dirOfRun(\rho)+(\vect{a}_i, \vect{b}_i)+\N n_i \vect{p}_i\) is flattable for every \(i\). We define the finite set \(\vectSet{F}:=\vect{b}+\{0,\dots, n_1\} \cdot \vect{p}_1+\dots+\{0,\dots, n_r\} \cdot \vect{p}_r \subseteq \vectSet{L} \subseteq \Component(\Expression,\rho)\). Since this set is finite, it is flattable. By property 1. in Theorem \ref{TheoremFlatteningLines}, the relation 

\(\sum_{i=1}^r \vectSet{R}_i + \vectSet{F}- r \cdot \dirOfRun(\rho)=\dirOfRun(\rho)+\sum_{i=1}^r [(\vect{a}_i, \vect{b}_i)+\N n_i \vect{p}_i]+(\vectSet{F}-\dirOfRun(\rho))\) is flattable. Defining \(\vect{p}:=\sum_{i=1}^r (\vect{a}_i, \vect{b}_i)+(\vect{b}-\dirOfRun(\rho)) \in \vectSet{P}_{\Expression, \rho}\) by periodicity, the theorem is proven.
\end{proof}

Next we outline how to prove Theorem \ref{TheoremFlatteningLines}. The first step is an equivalent characterization of the transformer relation via ``self-loop on \(\vect{c}\)'' relations similar to \(\VAS\). Let us start by defining these relations, which is similar to VAS. 

Consider an expression \(\Expression^{\ast}\) and a configuration \(\vect{c} \in \N^{d'}\). We define the relation \(\vectSet{P}_{\Expression^{\ast},\vect{c}}\) via \((\vect{x}, \vect{y}) \in \vectSet{P}_{\Expression^{\ast},\vect{c}} \iff \exists \rho \in \Omega(\Expression^{\ast}): \dirOfRun(\rho)=(\vect{x}+\vect{c},\vect{y}+\vect{c})\). 

The equivalent characterization of \(\vectSet{P}_{\Expression, \rho}\) via relations \(\vectSet{P}_{\Expression^{\ast},\vect{c}}\) is as follows:

\begin{lemma}
Let \(\Expression\) be an expression, and \(\rho \in \Omega(\Expression)\). Then there exist subexpressions \(\Expression_i^{\ast}\) of \(\Expression\) and configurations \(\vect{c}_i\) occurring along \(\rho\) such that \(\vectSet{P}_{\Expression, \rho}=\vectSet{P}_{\Expression_1^{\ast}, \vect{c}_1} \circ \dots \circ \vectSet{P}_{\Expression_r^{\ast}, \vect{c}_r}\). \label{LemmaEquivalentCharacterizationOfPRho}
\end{lemma}

This leaves us with three things to prove: Firstly, Lemma \ref{LemmaEquivalentCharacterizationOfPRho} itself. Secondly, that \(\vectSet{P}_{\Expression^{\ast}, \vect{c}}\) is smooth and properties 1. and 2. of Theorem \ref{TheoremFlatteningLines} hold for \(\vectSet{P}_{\Expression^{\ast}, \vect{c}}\) (actually a slightly stronger version 2.' of property 2.). Thirdly, that composition preserves the properties of Theorem \ref{TheoremFlatteningLines}. 

We dedicate one subsection to every step, with the second coming last. This is because steps 1 and 3 are new, while step 2 is based on \cite{Leroux13}. Also, in order to understand why PVAS are flattable, Lemma \ref{LemmaEquivalentCharacterizationOfPRho} and step 3 contain the essence: Similar to VAS, pumping is a \emph{sequence} of special self-loops, a very linear object. The fact that the different parts now use different expressions \(\Expression_i^{\ast}\) is irrelevant for our composition proof.

\subsection{Proving Lemma \ref{LemmaEquivalentCharacterizationOfPRho}, Equivalent Definition of Transformer Relation}

Proof by structural induction. For simplicity, we call a relation \(\vectSet{P}_{\Expression, \rho}\) decomposable if it has an equivalent description as in the lemma. The base case of VAS sections is clear.

\(\Expression_1 \cup \Expression_2\): Let \(\rho \in \Omega(\Expression_1 \cup \Expression_2)\). W.l.o.g. \(\rho \in \Omega(\Expression_1)\). By definition of the wqo on runs in this case, we have \(\vectSet{P}_{\Expression_1 \cup \Expression_2, \rho}=\vectSet{P}_{\Expression_1, \rho}\), which is decomposable by induction.

\(\Expression_1 \circ \Expression_2\): Let \(\rho \in \Omega(\Expression_1 \circ \Expression_2)\). Write \(\rho=\rho_1 \rho_2\) with \(\rho_1 \in \Omega(\Expression_1)\), \(\rho_2 \in \Omega(\Expression_2)\), which is a unique split because of the tags on steps of the run. We claim that \(\vectSet{P}_{\Expression_1 \circ \Expression_2, \rho}=\vectSet{P}_{\Expression_1, \rho_1} \circ \vectSet{P}_{\Expression_2, \rho_2}\).

Proof of claim: ``\(\subseteq\)'': Let \(\rho' \geq \rho\). Then \(\rho'= \rho_1' \rho_2'\) such that \(\rho_1' \in \Omega(\Expression_1), \rho_2' \in \Omega(\Expression_2), \rho_1' \geq_{\Expression_1} \rho_1,\rho_2' \geq_{\Expression_2} \rho_2\). Then \(\dirOfRun(\rho')-\dirOfRun(\rho)=(\dirOfRun(\rho_1')-\dirOfRun(\rho_1)) \circ (\dirOfRun(\rho_2')-\dirOfRun(\rho_2)) \in \vectSet{P}_{\Expression_1, \rho_1} \circ \vectSet{P}_{\Expression_2, \rho_2}\), where composition is possible since \(\target(\rho_1')=\source(\rho_2'), \target(\rho_1)=\source(\rho_2)\). 

The other direction ``\(\supseteq\)'' is clear, namely concatenate \(\rho_1'\) and \(\rho_2'\). 

Now simply use that both \(\vectSet{P}_{\Expression_1, \rho_1}\) and \(\vect{P}_{\Expression_2, \rho_2}\) are decomposable by induction.

This leaves the hardest case \(\Expression^{\ast}\): For configurations \(\vect{c}, \vect{c}' \in \N^{d'}\), we write \(\vect{c} \to_{\Expression} \vect{c}'\) if there exists a run \(\eta \in \Omega(\Expression)\) such that \(\dirOfRun(\eta)=(\vect{c}, \vect{c}')\), and call \(\eta\) a \emph{generalized transition}. We want to emphasize the important fact that generalized transitions contain path information in \(\Omega(\Expression)\), and are not only an element of \(\Rel(\Expression)\). We let \(\to_{\Expression}^{\ast}\) denote its reflexive and transitive closure. A first decomposition in this case contains relations \(\vectSet{P}_{\Expression, \eta_i}\), which correspond to ``increasing existing transitions'' as mentioned in the introduction. This leads to writing \(\vectSet{P}_{\Expression^{\ast}, \rho}\) as a composition of alternating \(\vectSet{P}_{\Expression^{\ast}, \vect{c}_i}\) and \(\vectSet{P}_{\Expression, \eta_i}\).

\begin{restatable}{lemma}{LemmaPumpingAsInJancar}
Let \(\rho=(\source(\rho), \eta_1 \dots \eta_r, \target(\rho)) \in \Omega(\Expression^{\ast})\). The following equality holds: \(\vectSet{P}_{\Expression^{\ast}, \rho}= \vectSet{P}_{\Expression^{\ast}, \source(\eta_1)} \circ \vectSet{P}_{\Expression, \eta_1} \circ \vectSet{P}_{\Expression^{\ast}, \source(\eta_2)} \circ \dots \circ \vectSet{P}_{\Expression, \eta_r} \circ \vectSet{P}_{\Expression^{\ast}, \target(\rho)}\).\label{LemmaPumpingAsInJancar}
\end{restatable}

\begin{proof}
``\(\Rightarrow\)'': Let \(\rho' \geq_{\Omega(\Expression^{\ast})} \rho\). We have to prove \(\dirOfRun(\rho')-\dirOfRun(\rho) \in \) the claimed composition. Write \(\rho'=\eta_1' \dots \eta_s'\) according to the tags. By definition of the wqo, there exists an order-preserving injective function \(f \colon \{1,\dots, r\} \to \{1,\dots,s\}\) such that \(\eta_i \leq_{\Omega(\Expression)} \eta_{f(i)}'\). In particular, \(\dirOfRun(\eta_{f(i)}') \geq \dirOfRun(\eta_i)\). We define the sequence of vectors \(\vect{v}_i:=\source(\eta_{f(i)}')-\source(\eta_i)\) for \(i \in \{1,\dots, r\}\), and \(\vect{w}_{i}:=\target(\eta_{f(i)}')-\target(\eta_i)\) for \(i \in \{1,\dots, r\}\). We also define \(\vect{w}_0:=\source(\rho')-\source(\rho)\) and \(\vect{v}_{r+1}:=\target(\rho')-\target(\rho)\). For every \(i \in \{1,\dots, r\}\), the run \(\eta_{f(i)}'\) shows that 
\((\vect{v}_i, \vect{w}_i) \in \vectSet{P}_{\Expression, \eta_i}\). The runs \(\rho_i:=\eta_{f(i)+1}\dots \eta_{f(i+1)-1}\) for \(i \in \{1,\dots, r-1\}\) (these are possibly empty runs) prove that \(\target(\eta_i)+ \vect{w}_i \to_{\Expression}^{\ast} \source(\eta_{i+1})+\vect{v}_{i+1}\). Since \(\target(\eta_i)=\source(\eta_{i+1})\), we obtain \((\vect{w}_i, \vect{v}_{i+1}) \in \vectSet{P}_{\Expression^{\ast}, \target(\eta_i)}\). A similar argument proves \((\vect{w}_0, \vect{v}_1) \in \vectSet{P}_{\Expression^{\ast}, \source(\rho)}\) and \((\vect{w}_r, \vect{v}_{r+1}) \in \vectSet{P}_{\Expression^{\ast}, \target(\rho)}\). Hence \(\dirOfRun(\rho')-\dirOfRun(\rho)=(\vect{w}_0, \vect{v}_{r+1})\) is in the claimed composition.

``\(\Leftarrow\)'': Let \((\vect{w}_i)_{i=0}^r\) and \((\vect{v}_i)_{i=1}^{r+1}\) such that \((\vect{w}_i, \vect{v}_{i+1}) \in \vectSet{P}_{\Expression, \eta_i}\) for \(i \in \{1,\dots, r\}\), \((\vect{v}_i, \vect{w}_i) \in \vectSet{P}_{\Expression^{\ast}, \target(\eta_i)}\) for \(i \in \{1,\dots, r-1\}\), \((\vect{w}_0, \vect{v}_1) \in \vectSet{P}_{\Expression^{\ast}, \source(\rho)}\) and \((\vect{w}_r, \vect{v}_{r+1}) \in \vectSet{P}_{\Expression^{\ast}, \target(\rho)}\). Let \(\eta_i'\) be generalized transitions witnessing \((\vect{w}_i, \vect{v}_{i+1}) \in \vectSet{P}_{\Expression, \eta_i}\), let \(\rho_i\) for \(i \in \{1,\dots, r-1\}\) be runs witnessing \((\vect{v}_i, \vect{w}_i) \in \vectSet{P}_{\Expression^{\ast}, \target(\eta_i)}\), let \(\rho_0\) witness \((\vect{w}_0, \vect{v}_1) \in \vectSet{P}_{\Expression^{\ast},\source(\rho)}\) and \(\rho_{r+1}\) witness \((\vect{w}_r, \vect{v}_{r+1}) \in \vectSet{P}_{\Expression^{\ast}, \target(\rho)}\). Then \(\rho':=\rho_0 \eta_1' \rho_1 \dots \rho_{r-1} \eta_r' \rho_r\) fulfills \(\rho' \in \Omega(\Expression^{\ast})\) because sources and targets of the different parts coincide. In fact \(\rho' \geq_{\Omega(\Expression^{\ast})} \rho\), by choosing \(f(i)\) to point at the index at which \(\eta_i'\) occurs in \(\rho'\).
\end{proof}

This finishes proving Lemma \ref{LemmaEquivalentCharacterizationOfPRho} by observing that \(\vectSet{P}_{\Expression, \eta_i}\) are decomposable by induction.

We remark that the proof does obtain an explicit description of which \(\vectSet{P}_{\Expression_i^{\ast}, \vect{c}_i}\) to use, but we stated Lemma \ref{LemmaEquivalentCharacterizationOfPRho} this way to stress that different \(\Expression_i^{\ast}\) intertwine in the composition.

\subsection{Preserving Smoothness and Flattability Under Composition} \label{SectionSmoothnessComposition}
In this subsection we prove that if \(\vectSet{P}_{\Expression, \vect{c}}\) are smooth periodic relations fulfilling properties 1. and 2. of Theorem \ref{TheoremFlatteningLines}, then also their composition fulfills these conditions. While periodic relations are closed under composition (use e.g. Lemma \ref{LemmaExchangePlusAndComposition}), smooth periodic relations are not. We already slightly changed the definition of well-directed to accomplish this goal (compare with \cite{Leroux13, GuttenbergRE23}), but we still need to carefully choose the inductive statement. We choose to replace condition 2. by 2.' formulated as follows, which is similar to \cite{Leroux13}:

2.': For every well-directed periodic \(P \subseteq \vectSet{P}\) there exists a definable cone \(R\) such that \(\dir(P) \subseteq R\) and for every \((\vect{e}, \vect{f}) \in R\) there exist \((\vect{a}, \vect{b}) \in P, n \in \N\) such that \((\vect{a}, \vect{b})+\N n(\vect{e}, \vect{f}) \subseteq \vectSet{P}\). In case of \(\vectSet{P}_{\Expression, \rho}\) we require \(\dirOfRun(\rho)+(\vect{a}, \vect{b})+\N n (\vect{e}, \vect{f}) \subseteq \Component(\Expression,\rho)\) to be flattable.

The idea of property 2.' is to remove one basic difficulty of composition: Suddenly not all runs of \(\vectSet{P}_1\) are useful anymore, only those which can be continued into \(\vectSet{P}_2\). We will see in the proof how 2' takes care of this problem. Formally, property 2' says that even if \(P \subseteq \vectSet{P}\) is non-smooth, we can find a definable \(R\) with \(\dir(P) \subseteq R \subseteq \dir(\vectSet{P})\), and it only contains flattable directions. With the choice \(P=\vectSet{P}\) this implies property 2. Important to notice is that the choice \(R=\dir(\vectSet{P})\) would always be best if not for the very crucial \(\in P\) part, which we will use in the proof. Again, \(\in P\) is easy to motivate. Imagine one is interested in the reachability set from a fixed point, i.e. in only pumping the target. Then choose \(P:=\{\vect{0}\} \times \N^d\). Property 2 would state existence of a line \((\vect{a},\vect{b})+\N (0, \vect{w}) \subseteq \vectSet{P}_{\Expression, \rho}\). We actually want a line \((\vect{0}, \vect{b})+\N (\vect{0}, \vect{w})\), i.e. with \((\vect{a}, \vect{b}) \in P\) as in property 2'.

\begin{lemma}
Let \(\vectSet{P}_1, \vectSet{P}_2\) be smooth periodic relations fulfilling property 2.'. Then \(\vectSet{P}_1 \circ \vectSet{P}_2\) is smooth periodic fulfilling property 2'. If \(\vectSet{P}_1=\vectSet{P}_{\Expression_1, \rho_1}\) and \(\vectSet{P}_2=\vectSet{P}_{\Expression_2, \rho_2}\) for \(\rho=\rho_1 \rho_2\), then in addition the flattability claims of property 1. and 2.' hold for \(\vectSet{P}_{\Expression_1 \circ \Expression_2, \rho}=\vectSet{P}_{\Expression_1, \rho_1} \circ \vectSet{P}_{\Expression_2, \rho_2}\). \label{LemmaCompositionSmooth}
\end{lemma}

\begin{proof}
Periodic: A composition of periodic relations is again periodic by Lemma \ref{LemmaExchangePlusAndComposition}.

Well-directed: Let \((\vect{v}_n, \vect{w}_n)_n \subseteq \vectSet{P}_1 \circ \vectSet{P}_2\) be a sequence. Then there exist intermediate values \(\vect{x}_n\) such that \((\vect{v}_n, \vect{x}_n) \in \vectSet{P}_1, (\vect{x}_n, \vect{w}_n) \in \vectSet{P}_2\) for all \(n\). Since \(\vectSet{P}_1\) is well-directed, there exists a subsequence such that \((\vect{v}_{n_j}, \vect{x}_{n_j})+\N (\vect{v}_{n_k}-\vect{v}_{n_j}, \vect{x}_{n_k}-\vect{x}_{n_j}) \subseteq \vectSet{P}_1\). Since \(\vectSet{P}_2\) is well-directed we obtain additionally \((\vect{x}_{n_j}, \vect{w}_{n_j})+\N (\vect{x}_{n_k}-\vect{x}_{n_j}, \vect{w}_{n_k}-\vect{w}_{n_j}) \subseteq \vectSet{P}_2\) for some subsubsequence. Together we have \((\vect{v}_{n_j}, \vect{w}_{n_j})+\N (\vect{v}_{n_k}-\vect{v}_{n_j}, \vect{w}_{n_k}-\vect{w}_{n_j}) \subseteq \vectSet{P}_1 \circ \vectSet{P}_2\).

Property 2.': Let \(P \subseteq \vectSet{P}_1 \circ \vectSet{P}_2\) be well-directed periodic. Define \(P':=\{(\vect{v}, \vect{x}, \vect{w}) \mid (\vect{v}, \vect{w}) \in P, (\vect{v}, \vect{x}) \in \vectSet{P}_1, (\vect{x}, \vect{w}) \in \vectSet{P}_2\}\). \(P'\) is well-directed by an argument as above, but this time we even have to choose a subsubsubsequence. Consider the projections \(\pi_{12}\) and \(\pi_{23}\) to \((\vect{v}, \vect{x})\) and \((\vect{x}, \vect{w})\) respectively. \(P_1:=\pi_{12}(P') \subseteq \vectSet{P}_1\) and \(P_2:=\pi_{23}(P') \subseteq \vectSet{P}_2\) are projections of a well-directed periodic relation and therefore themselves well-directed periodic. Hence we can apply property 2.' for them to obtain definable cones \(R_1\) and \(R_2\) with \(\dir(P_i) \subseteq R_i \subseteq \dir(\vectSet{P}_i)\). We claim property 2' holds with \(R=R_1 \circ R_2\). 

First we have to show that \(\dir(P) \subseteq R_1 \circ R_2\). Let \((\vect{v}, \vect{w}) \in \dir(P)\). By definition of the set of directions, by potentially scaling with a positive integer, there exists \((\vect{v}_0, \vect{w}_0)\) such that \((\vect{v}_n, \vect{w}_n):=(\vect{v}_0, \vect{w}_0)+n (\vect{v}, \vect{w}) \in \vectSet{P}_1 \circ \vectSet{P}_2\) for every \(n\). Therefore there exist intermediate values \(\vect{x}_n\) such that \((\vect{v}_n, \vect{x}_n, \vect{w}_n) \in P'\) for all \(n\). Since \(P'\) is well-directed, there exists a subsequence such that \((\vect{v}_{n_j}, \vect{x}_{n_j}, \vect{w}_{n_j})+\N (\vect{v}_{n_k}-\vect{v}_{n_j}, \vect{x}_{n_k}-\vect{x}_{n_j}, \vect{w}_{n_k}-\vect{w}_{n_j}) \subseteq P'\). Hence \((\vect{v}_{n_k}-\vect{v}_{n_j}, \vect{x}_{n_k}-\vect{x}_{n_j}) \in \dir(P_1) \subseteq R_1\) and \((\vect{x}_{n_k}-\vect{x}_{n_j}, \vect{w}_{n_k}-\vect{w}_{n_j}) \in \dir(P_2) \subseteq R_2\). Therefore \((\vect{v}_{n_k}-\vect{v}_{n_j}, \vect{w}_{n_k}-\vect{w}_{n_j})=(n_k-n_j)(\vect{v}, \vect{w}) \in R_1 \circ R_2\). This implies \((\vect{v}, \vect{w}) \in R_1 \circ R_2\).

Now let \((\vect{e}, \vect{g}) \in R\). Then there exists \(\vect{f}\) such that \((\vect{e}, \vect{f}) \in R_1\) and \((\vect{f}, \vect{g}) \in R_2\). By definition of \(R_i\), by potentially scaling, there exist \((\vect{a}_1, \vect{b}_1) \in P_1\) and \((\vect{b}_2, \vect{c}_2) \in P_2\) such that \((\vect{a}_1, \vect{b}_1)+\N (\vect{e}, \vect{f}) \subseteq \vectSet{P}_1\), and \((\vect{b}_2, \vect{c}_2)+\N (\vect{f}, \vect{g}) \subseteq \vectSet{P}_2\). Since \(P_1\) and \(P_2\) are projections of \(P'\), there exist \(\vect{c}_1\) and \(\vect{a}_2\) such that \((\vect{a}_i, \vect{b}_i, \vect{c}_i) \in P'\) for \(i \in \{1,2\}\). Hence \((\vect{a}_2, \vect{b}_2) \in P_1\) and \((\vect{b}_1, \vect{c}_1) \in P_2\). By periodicity of \(\vectSet{P}_1\), we have \((\vect{a}_1+\vect{a}_2, \vect{b}_1+\vect{b}_2)+ \N (\vect{e}, \vect{f}) \subseteq \vectSet{P}_1\) and similarly by periodicity of \(\vectSet{P}_2\) we have \((\vect{b}_1+\vect{b}_2, \vect{c}_1+\vect{c}_2)+\N (\vect{f}, \vect{g}) \subseteq \vectSet{P}_2\). Altogether we have \((\vect{a}_1+\vect{a}_2, \vect{c}_1+\vect{c}_2) \in P\) and \((\vect{a}_1+\vect{a}_2, \vect{c}_1+\vect{c}_2)+\N (\vect{e}, \vect{g}) \subseteq \vectSet{P}_1 \circ \vectSet{P}_2\) as required. \((\vect{a}_1, \vect{b}_1) \in P_1\) was crucial here such that we could obtain a fitting \(\vect{c}_1\), and accordingly for \((\vect{b}_2, \vect{c}_2) \in P_2\).

Asymptotically definable: Define \(P:=\vectSet{P}_1 \circ \vectSet{P}_2\). By property 2', we have \(\dir(\vectSet{P}_1 \circ \vectSet{P}_2)=R_1 \circ R_2\), which as composition of definable cones is itself a definable cone.

Flattability claim in 2': We proved 2' and constructed the direction using well-directedness. By Lemma \ref{LemmaBasicflattabilityProperties}(1.), relations obtained this way from the wqo. on runs are flattable.

Property 1.: Let \(\vectSet{R}, \vectSet{R}' \subseteq \dirOfRun(\rho)+\vectSet{P}_{\Expression_1 \circ \Expression_2, \rho}\) be flattable w.r.t. \(\Expression_1 \circ \Expression_2\). We have to prove that \(\vectSet{R}+\vectSet{R}'-\dirOfRun(\rho)\) is flattable w.r.t. \(\Expression_1 \circ \Expression_2\). By definition of flattable, there exist relations \(\vectSet{R}_1, \vectSet{R}_1'\) flattable w.r.t. \(\Expression_1\) and \(\vectSet{R}_2, \vectSet{R}_2'\) flattable w.r.t. \(\Expression_2\) such that \(\vectSet{R} \subseteq \vectSet{R}_1 \circ \vectSet{R}_2\) and \(\vectSet{R}' \subseteq \vectSet{R}_1' \circ \vectSet{R}_2'\). Hence \(\vectSet{R}+ \vectSet{R}'-\dirOfRun(\rho) \subseteq (\vectSet{R}_1 \circ \vectSet{R}_2) + (\vectSet{R}_1' \circ \vectSet{R}_2')-(\dirOfRun(\rho_1) \circ \dirOfRun(\rho_2))\). By Lemma \ref{LemmaExchangePlusAndComposition} we obtain \(\vectSet{R}+\vectSet{R}'-\dirOfRun(\rho) \subseteq (\vectSet{R}_1+\vectSet{R}_1'- \dirOfRun(\rho_1)) \circ (\vectSet{R}_2+\vectSet{R}_2'- \dirOfRun(\rho_2))\). Applying property 1. for the subexpressions \(\Expression_1\) and \(\Expression_2\), we obtain the claim.
\end{proof}

\subsection{Transformer Relations are Smooth with Flattable Directions} \label{SubsectionEStarSmooth}

In this subsection we prove that the transformer relation \(\vectSet{P}_{\Expression^{\ast}, \vect{c}}\) is a smooth periodic relation fulfilling properties 1 and 2' (see Section \ref{SectionSmoothnessComposition}). We start with a reminder of the notation.

We write \(\vect{x} \to_{\Expression} \vect{y}\) for \((\vect{x}, \vect{y}) \in \Rel(\Expression)\) and denote its reflexive transitive closure as \(\to_{\Expression}^{\ast}\). Remember that \((\vect{x}, \vect{y}) \in \vectSet{P}_{\Expression^{\ast}, \vect{c}} \iff \vect{c}+\vect{x} \to_{\Expression}^{\ast} \vect{c}+\vect{y}\). We first prove that \(\vectSet{P}_{\Expression^{\ast}, \vect{c}}\) is periodic.

\begin{lemma}
Let \(\vect{c}+\vect{x}_1 \to_{\Expression}^{\ast} \vect{c}+ \vect{y}_1\) and \(\vect{c}+\vect{x}_2 \to_{\Expression}^{\ast} \vect{c}+\vect{y}_2\). Then \(\vect{c}+\vect{x}_1+\vect{x}_2 \to_{\Expression}^{\ast} \vect{c}+\vect{y}_1+ \vect{y}_2\).\label{LemmaPCPeriodic}
\end{lemma}

\begin{proof}
Since \(\ast\) in expressions is only used on monotone relations \(\Rel(\Expression)\), \(\to_{\Expression}\) is monotone. Hence also \(\to_{\Expression}^{\ast}\) is monotone. By monotonicity, \(\vect{c}+\vect{x}_1+\vect{x}_2 \to_{\Expression}^{\ast} \vect{c}+\vect{y}_1+\vect{x}_2 \to_{\Expression}^{\ast} \vect{c}+\vect{y}_1+\vect{y}_2\).
\end{proof}

That \(\vectSet{P}_{\Expression^{\ast}, \vect{c}}\) is well-directed follows from Lemma \ref{LemmaBasicflattabilityProperties} (1.). 

The proof of Property 1. is similar to the above proof of Lemma \ref{LemmaPCPeriodic}. Indeed, the lemma did not duplicate \(\vect{c}\), so in different notation \(\vect{r}_1 \to_{\Expression}^{\ast} \vect{s}_1\) and \(\vect{r}_2 \to_{\Expression}^{\ast} \vect{s}_2\) imply \(\vect{r}_1+\vect{r}_2-\vect{c} \to_{\Expression}^{\ast} \vect{s}_1+\vect{s}_2-\vect{c}\). This leaves proving that \(\vectSet{P}_{\Expression^{\ast},\vect{c}}\) fulfills property 2.' and hence is asymptotically definable.

Since \((\vect{0}, \vect{0}) \in \vectSet{P}_{\Expression^{\ast}, \vect{c}}\), monotonicity implies that \(\vectSet{P}_{\Expression^{\ast}, \vect{c}}\) is reflexive. In \cite[Section VIII]{Leroux13} many lemmas were proven for reflexive periodic relations, we reuse many of them. This leads to a long sequence of restating lemmas, we prefer to end this part of the main text by sketching the idea, for the formal proof see Appendix \ref{SubsectionAppendixTStarCase}.

Let \(P \subseteq \vectSet{P}_{\Expression^{\ast}, \vect{c}}\) periodic. In order to not confuse \(\Omega(\Expression)\) and \(\Omega(\Expression^{\ast})\), we write \(\eta \in \Omega(\Expression)\) and \(\rho \in \Omega(\Expression^{\ast})\). We call \(\eta\) a \emph{generalized transition}. Write \(\gamma=(\Expression^{\ast}, \vect{c}, P)\) and denote by \(\Omega_{\gamma}\) the set of runs \(\rho \in \Omega(\Expression^{\ast})\) such that \(\dirOfRun(\rho) \in (\vect{c}, \vect{c})+P\). The idea is to split counters into bounded and unbounded counters for \(\Omega_{\gamma}\). The bounded counters will all be stored in the states of a graph, and this leads to pumping corresponding to cycles in the graph. Namely any transition sequence will correspond to a path in the graph, and since bounded counters cannot be pumped, any pumping sequence has to restore all the bounded counters, i.e. be a cycle in the graph. The unbounded counters on the other hand will all be unbounded \emph{simultaneously}, at which point the condition that counters have to stay non-negative will intuitively not influence possible behaviours anymore.

We hence define the graph of bounded counters. Let \(Q_{\gamma} \subseteq \N^{d'}\) be the set of configurations occurring on some run \(\rho \in \Omega_{\gamma}\). We denote by \(I_{\gamma}\) the set of indices such that \(\{q(i) \mid q \in Q_{\gamma}\}\) is finite, i.e. the set of bounded counters. We consider the projection \(\pi_{\gamma} \colon \N^{d'} \to \N^{I_{\gamma}}\) to the bounded counters. We now define a finite directed multigraph \(G_{\gamma}\) with vertices \(S_{\gamma}:=\pi_{\gamma}(Q_{\gamma})\). For edges \((s,t)\), first consider the set \(\Omega_{s,t}\) of generalized transitions \(\eta\) with \(\pi_{\gamma}(\source(\eta))=s\) and \(\pi_{\gamma}(\target(\eta))=t\), which occur in some run \(\rho \in \Omega_{\gamma}\). We add an edge \((s,t)\) for every minimal (w.r.t. \(\leq_{\Omega(\Expression)}\)) element of \(\Omega_{s,t}\). We let \(s_{\gamma}=\pi_{\gamma}(\vect{c})\) denote the ``initial state'' for this graph.

Clearly runs correspond to paths in \(G_{\gamma}\), since the graph has projections of configurations as states, and generalized transitions as edges. Regarding unbounded counters, the proof that all of them are unbounded \emph{simultaneously} is in Appendix \ref{SubsectionAppendixTStarCase}. Finally, we then obtain a formula for a definable cone \(R\) overapproximating \(\dir(P)\) by first considering the finitely many minimal cycles in \(G_{\gamma}\). Every cycle \(\eta_1 \dots \eta_m\) provides us with a smooth periodic relation \(\vectSet{P}_{\Expression, \eta_1} \circ \dots \circ \vectSet{P}_{\Expression, \eta_m}\) of pumping possible along this cycle. These relations are smooth periodic by induction and Lemmas \ref{LemmaEquivalentCharacterizationOfPRho} and \ref{LemmaCompositionSmooth}. We conclude using \cite[Theorem VII.1]{Leroux13}, whose statement is essentially the following: If \(\vectSet{P}_1, \dots, \vectSet{P}_m\) are \emph{reflexive} asymptotically definable periodic relations, then \((\bigcup_{i=1}^m \vectSet{P}_i)^{\ast}\) is also asymptotically definable.

\section{Conclusion} \label{SectionConclusion}

% !TeX root = Main.tex

We have given a new characterization of PVAS sections as RegEx over VAS sections, and extended the abstract properties of almost semilinear sets to PVAS sections. We have concluded that therefore if the reachability relation of a PVAS is semilinear, then it is flattable, and moreover if a configuration is not reachable, then it is separated by a semilinear inductive invariant. This leaves two main unknowns for PVAS which are known for VAS: 1) The decidability of the semilinearity problem, that is, given a PVAS, decide if its reachability relation/set is semilinear. 2) The complexity of the reachability problem.

We leave these as future work. We believe that combining our characterization of PVAS sections and wqo on runs with ideas from \cite{GuttenbergRE23} might allow for progress on these open problems.

Furthermore, this research can also be viewed as progress towards the open question whether the reachability problem for Pushdown/Grammar VASS is decidable, where Pushdown VASS have a stack in addition to the counters. Namely it is known \cite{AtigG11} that Priority VASS are equivalent to a subclass of Pushdown VASS.

\bibliography{Main.bib}

\begin{thebibliography}{10}

\bibitem{AtigG11}
Mohamed~Faouzi Atig and Pierre Ganty.
\newblock Approximating petri net reachability along context-free traces.
\newblock In {\em {FSTTCS}}, volume~13 of {\em LIPIcs}, pages 152--163. Schloss
  Dagstuhl - Leibniz-Zentrum f{\"{u}}r Informatik, 2011.

\bibitem{BlondinL23}
Michael Blondin and Fran{\c{c}}ois Ladouceur.
\newblock Population protocols with unordered data.
\newblock In {\em {ICALP}}, volume 261 of {\em LIPIcs}, pages 115:1--115:20.
  Schloss Dagstuhl - Leibniz-Zentrum f{\"{u}}r Informatik, 2023.

\bibitem{Bonnet11}
R{\'{e}}mi Bonnet.
\newblock The reachability problem for vector addition system with one
  zero-test.
\newblock In {\em {MFCS}}, volume 6907 of {\em Lecture Notes in Computer
  Science}, pages 145--157. Springer, 2011.

\bibitem{Bonnet12}
R{\'e}mi Bonnet.
\newblock {\em Theory of Well-Structured Transition Systems and Extended
  Vector-Addition Systems}.
\newblock Th{\`e}se de doctorat, Laboratoire Sp{\'e}cification et
  V{\'e}rification, ENS Cachan, France, 2013.
\newblock URL:
  \url{http://www.lsv.ens-cachan.fr/Publis/PAPERS/PDF/bonnet-phd13.pdf}.

\bibitem{ClementeCLP17}
Lorenzo Clemente, Wojciech Czerwinski, Slawomir Lasota, and Charles Paperman.
\newblock Separability of reachability sets of vector addition systems.
\newblock In {\em {STACS}}, volume~66 of {\em LIPIcs}, pages 24:1--24:14.
  Schloss Dagstuhl - Leibniz-Zentrum f{\"{u}}r Informatik, 2017.

\bibitem{CzerwinskiLLLM19}
Wojciech Czerwinski, Slawomir Lasota, Ranko Lazic, J{\'{e}}r{\^{o}}me Leroux,
  and Filip Mazowiecki.
\newblock The reachability problem for petri nets is not elementary.
\newblock In {\em {STOC}}, pages 24--33. {ACM}, 2019.

\bibitem{CzerwinskiO21}
Wojciech Czerwinski and Lukasz Orlikowski.
\newblock Reachability in vector addition systems is ackermann-complete.
\newblock In {\em {FOCS}}, pages 1229--1240. {IEEE}, 2021.

\bibitem{FinkelLS18}
Alain Finkel, J{\'{e}}r{\^{o}}me Leroux, and Gr{\'{e}}goire Sutre.
\newblock Reachability for two-counter machines with one test and one reset.
\newblock In {\em {FSTTCS}}, volume 122 of {\em LIPIcs}, pages 31:1--31:14.
  Schloss Dagstuhl - Leibniz-Zentrum f{\"{u}}r Informatik, 2018.

\bibitem{GuttenbergRE23}
Roland Guttenberg, Mikhail~A. Raskin, and Javier Esparza.
\newblock Geometry of reachability sets of vector addition systems.
\newblock In {\em {CONCUR}}, volume 279 of {\em LIPIcs}, pages 6:1--6:16.
  Schloss Dagstuhl - Leibniz-Zentrum f{\"{u}}r Informatik, 2023.

\bibitem{Haase18}
Christoph Haase.
\newblock A survival guide to presburger arithmetic.
\newblock {\em {ACM} {SIGLOG} News}, 5(3):67--82, 2018.

\bibitem{HaaseZ19}
Christoph Haase and Georg Zetzsche.
\newblock Presburger arithmetic with stars, rational subsets of graph groups,
  and nested zero tests.
\newblock In {\em {LICS}}, pages 1--14. {IEEE}, 2019.

\bibitem{HofmanLLLST16}
Piotr Hofman, Slawomir Lasota, Ranko Lazic, J{\'{e}}r{\^{o}}me Leroux, Sylvain
  Schmitz, and Patrick Totzke.
\newblock Coverability trees for petri nets with unordered data.
\newblock In {\em FoSSaCS}, volume 9634 of {\em Lecture Notes in Computer
  Science}, pages 445--461. Springer, 2016.

\bibitem{HopcroftP79}
John~E. Hopcroft and Jean{-}Jacques Pansiot.
\newblock On the reachability problem for 5-dimensional vector addition
  systems.
\newblock {\em Theor. Comput. Sci.}, 8:135--159, 1979.

\bibitem{Jancar90}
Petr Jancar.
\newblock Decidability of a temporal logic problem for petri nets.
\newblock {\em Theor. Comput. Sci.}, 74(1):71--93, 1990.

\bibitem{Kosaraju82}
S.~Rao Kosaraju.
\newblock Decidability of reachability in vector addition systems.
\newblock In {\em {STOC}}, pages 267--281. {ACM}, 1982.

\bibitem{Lambert92}
Jean{-}Luc Lambert.
\newblock A structure to decide reachability in petri nets.
\newblock {\em Theor. Comput. Sci.}, 99(1):79--104, 1992.

\bibitem{LazicS16}
Ranko Lazic and Sylvain Schmitz.
\newblock The complexity of coverability in {\(\nu\)}-petri nets.
\newblock In {\em {LICS}}, pages 467--476. {ACM}, 2016.

\bibitem{Leroux09}
J{\'{e}}r{\^{o}}me Leroux.
\newblock The general vector addition system reachability problem by presburger
  inductive invariants.
\newblock In {\em {LICS}}, pages 4--13. {IEEE} Computer Society, 2009.

\bibitem{Leroux11}
J{\'{e}}r{\^{o}}me Leroux.
\newblock Vector addition system reachability problem: {A} short self-contained
  proof.
\newblock In {\em {LATA}}, volume 6638 of {\em Lecture Notes in Computer
  Science}, pages 41--64. Springer, 2011.

\bibitem{Leroux12}
J{\'{e}}r{\^{o}}me Leroux.
\newblock Vector addition systems reachability problem {(A} simpler solution).
\newblock In {\em Turing-100}, volume~10 of {\em EPiC Series in Computing},
  pages 214--228. EasyChair, 2012.

\bibitem{Leroux13}
J{\'{e}}r{\^{o}}me Leroux.
\newblock Presburger vector addition systems.
\newblock In {\em {LICS}}, pages 23--32. {IEEE} Computer Society, 2013.
\newblock URL: \url{https://hal.science/hal-00780462v2}.

\bibitem{Leroux21}
J{\'{e}}r{\^{o}}me Leroux.
\newblock The reachability problem for petri nets is not primitive recursive.
\newblock In {\em {FOCS}}, pages 1241--1252. {IEEE}, 2021.

\bibitem{LerouxPSS19}
J{\'{e}}r{\^{o}}me Leroux, M.~Praveen, Philippe Schnoebelen, and Gr{\'{e}}goire
  Sutre.
\newblock On functions weakly computable by pushdown petri nets and related
  systems.
\newblock {\em Log. Methods Comput. Sci.}, 15(4), 2019.

\bibitem{LerouxPS14}
J{\'{e}}r{\^{o}}me Leroux, M.~Praveen, and Gr{\'{e}}goire Sutre.
\newblock Hyper-ackermannian bounds for pushdown vector addition systems.
\newblock In {\em {CSL-LICS}}, pages 63:1--63:10. {ACM}, 2014.

\bibitem{LerouxS20}
J{\'{e}}r{\^{o}}me Leroux and Gr{\'{e}}goire Sutre.
\newblock Reachability in two-dimensional vector addition systems with states:
  One test is for free.
\newblock In {\em {CONCUR}}, volume 171 of {\em LIPIcs}, pages 37:1--37:17.
  Schloss Dagstuhl - Leibniz-Zentrum f{\"{u}}r Informatik, 2020.

\bibitem{LerouxST15}
J{\'{e}}r{\^{o}}me Leroux, Gr{\'{e}}goire Sutre, and Patrick Totzke.
\newblock On the coverability problem for pushdown vector addition systems in
  one dimension.
\newblock In {\em {ICALP} {(2)}}, volume 9135 of {\em Lecture Notes in Computer
  Science}, pages 324--336. Springer, 2015.

\bibitem{Mayr81}
Ernst~W. Mayr.
\newblock An algorithm for the general petri net reachability problem.
\newblock In {\em {STOC}}, pages 238--246. {ACM}, 1981.

\bibitem{PiskacK08}
Ruzica Piskac and Viktor Kuncak.
\newblock Linear arithmetic with stars.
\newblock In {\em {CAV}}, volume 5123 of {\em Lecture Notes in Computer
  Science}, pages 268--280. Springer, 2008.

\bibitem{Reinhardt08}
Klaus Reinhardt.
\newblock Reachability in petri nets with inhibitor arcs.
\newblock In {\em {RP}}, volume 223 of {\em Electronic Notes in Theoretical
  Computer Science}, pages 239--264. Elsevier, 2008.

\bibitem{RosaVelardoF11}
Fernando Rosa{-}Velardo and David de~Frutos{-}Escrig.
\newblock Decidability and complexity of petri nets with unordered data.
\newblock {\em Theor. Comput. Sci.}, 412(34):4439--4451, 2011.

\end{thebibliography}

\appendix

% !TeX root = Main.tex

\section{Proofs of Section \ref{SectionProofFinalTheorem}} \label{SectionAppendixProofFinalTheorem}

In this section we prove the lemmas of Section \ref{SectionProofFinalTheorem} in the order they are enumerated.

\LemmaBasicExpressionflattability*

\begin{proof}
We prove the properties one after the other, each by structural induction. 

Property 1.: Base case: See \cite[Lemma VI.3]{Leroux13}.

\(\Expression_1 \cup \Expression_2\) follows by induction, since \(\rho \leq_{\Omega(\Expression_1 \cup \Expression_2)} \rho'\) implies w.l.o.g. \(\rho \leq_{\Omega(\Expression_1)} \rho'\).

\(\Expression_1 \circ \Expression_2\): Write \(\rho'=\rho_1' \rho_2'\) and \(\rho=\rho_1 \rho_2\). Since \(\rho_i' \geq_{\Omega(\Expression_i)} \rho_i\) for \(i \in \{1,2\}\), by induction, \(\dirOfRun(\rho_i)+\N(\dirOfRun(\rho_i')-\dirOfRun(\rho_i))\) is flattable w.r.t. \(\Expression_i\). We have \(\dirOfRun(\rho)+\N(\dirOfRun(\rho')-\dirOfRun(\rho))\) \(\subseteq [\dirOfRun(\rho_1)+\N(\dirOfRun(\rho_1')-\dirOfRun(\rho_1))] \circ [\dirOfRun(\rho_2)+\N(\dirOfRun(\rho_2')-\dirOfRun(\rho_2))]\) by Lemma \ref{LemmaExchangePlusAndComposition}, and hence flattable w.r.t. \(\Expression_1 \circ \Expression_2\).

\(\Expression^{\ast}\): Let \(\rho' \geq_{\Omega(\Expression^{\ast})} \rho\). Write \(\rho'=\eta_1' \dots \eta_s'\) and \(\rho=\eta_1 \dots \eta_r\) according to the tags. By definition of the wqo, there exists an order-preserving injective function \(f \colon \{1,\dots, r\} \to \{1,\dots,s\}\) such that \(\eta_i \leq_{\Omega(\Expression)} \eta_{f(i)}\). We write \(\rho'=\rho_0 \eta_{f(1)}' \rho_1 \dots \rho_{r-1} \eta_{f(r)}' \rho_r\). By induction, for every \(i\) the relation \(\vectSet{R}_i:=\dirOfRun(\eta_i)+\N(\dirOfRun(\eta_{f(i)})-\dirOfRun(\eta_i))\) is flattable w.r.t. \(\Expression\). We define linear path schemes \((\vectSet{S}_i)_{i=0}^r\) by \(\vectSet{S}_i:=\mtc(\rho_i)\). The relation \(\dirOfRun(\rho)+\N(\dirOfRun(\rho')-\dirOfRun(\rho)) \subseteq \vectSet{S}_0 \circ \vectSet{R}_1 \circ \dots \circ \vectSet{R}_r \circ \vectSet{S}_r\) by arguments similar to the proof of Lemma \ref{LemmaPumpingAsInJancar}. Hence this relation is flattable w.r.t. \(\Expression^{\ast}\) by definition.

Property 2.: Base case: Follows from the closure properties of semilinear sets.

\(\Expression_1 \cup \Expression_2\): Let \(\vectSet{R}_1, \vectSet{R}_1'\) flattable w.r.t. \(\Expression_1\), \(\vectSet{R}_2,\vectSet{R}_2'\) flattable w.r.t. \(\Expression_2\) such that \(\vectSet{R} \subseteq \vectSet{R}_1 \cup \vectSet{R}_2\) and \(\vectSet{R}' \subseteq \vectSet{R}_1' \cup \vectSet{R}_2'\). Then \(\vectSet{R} \cup \vectSet{R}' \subseteq (\vectSet{R}_1 \cup \vectSet{R}_1') \cup (\vectSet{R}_2 \cup \vectSet{R}_2')\), finish by induction. 

\(\Expression_1 \circ \Expression_2\): By a similar exchanging of operations as Lemma \ref{LemmaExchangePlusAndComposition} with \(\circ, \cup\).

\(\Expression^{\ast}\): Let \(\vectSet{R} \subseteq Flat:=\vectSet{S}_0 \circ \vectSet{R}_1 \circ \dots \circ \vectSet{R}_r \circ \vectSet{S}_r\) and \(\vectSet{R}' \subseteq Flat':=\vectSet{S}_0' \circ \vectSet{R}_1' \circ \dots \circ \vectSet{R}_k' \circ \vectSet{S}_k'\) with linear path schemes \(\vectSet{S}_i, \vectSet{S}_i'\) and \(\vectSet{R}_i, \vectSet{R}_i'\) flattable w.r.t. \(\Expression\). Since every subexpression \(\Expression'\) of \(\Expression\) fulfills \(in(\Expression'), out(\Expression') \geq in(\Expression)\) (the assumption we guaranteed monotonicity with), every subexpression \(\Expression'\) is monotone in the \(j\)-th last counter for every \(j \in \{1,\dots, in(\Expression)\}\). Therefore we can assume that the \(\vectSet{R}_i, \vectSet{R}_i'\) are monotone in the \(j\)-th last counter for all \(j \in \{1,\dots, in(\Expression)\}\). By furthermore adding the point \((\vect{0}, \vect{0})\) (single points are always flattable, and by induction we can take the union), the \(\vectSet{R}_i, \vectSet{R}_i'\) are hence w.l.o.g. reflexive. Since moreover the linear path schemes \(\vectSet{S}_i, \vectSet{S}_i'\) are reflexive by definition and reflexive relations are closed under (same dimension) composition, the relations \(Flat, Flat'\) are reflexive. Hence \(Flat \circ Flat'\) contains \(Flat\) and \(Flat'\). In particular \(\vectSet{R} \cup \vectSet{R}' \subseteq Flat \circ Flat'\) is flattable w.r.t. \(\Expression^{\ast}\) by definition.

Property 3.: Base case: See \cite[Theorem XI.2]{Leroux13}.

\(\Expression_1 \cup \Expression_2\): Let \(\VAS=(Q,E)\) be the PVAS generating the expression. Let \(\vectSet{R} \subseteq \Rel(\Expression_1 \cup \Expression_2)\) be flattable. By definition \(\vectSet{R} \subseteq \vectSet{R}_1 \cup \vectSet{R}_2\) where \(\vectSet{R}_i\) is flattable w.r.t. \(\Expression_i\). By induction, there exist words \(w_1 \dots w_r \in E^{\ast}\) and \(w_1' \dots w_s' \in E^{\ast}\) such that \(\vectSet{R}_1 \subseteq \to_{w_1}^{\ast} \circ \dots \circ \to_{w_r}^{\ast}\) and \(\vectSet{R}_2 \subseteq \to_{w_1'}^{\ast} \circ \dots \circ \to_{w_s'}^{\ast}\). Then \(\vectSet{R}_1 \cup \vectSet{R}_2 \subseteq \to_{w_1}^{\ast} \circ \dots \circ \to_{w_r}^{\ast} \circ \to_{w_1'}^{\ast} \circ \dots \circ \to_{w_s'}^{\ast}\) is flat.

\(\Expression_1 \circ \Expression_2\): Use the same combination of words as above.

\(\Expression^{\ast}\): Let \(\vectSet{R} \subseteq \vectSet{S}_0 \circ \vectSet{R}_1 \circ \dots \vectSet{R}_r \circ \vectSet{S}_r\) be flattable w.r.t. \(\Expression^{\ast}\). Write \(\vectSet{S}_i=\mtc(\dirOfRun(\rho_{i,1})) \circ \dots \circ \mtc(\dirOfRun(\rho_{i, r_i}))\). Let \(E(\rho_{i,j})\) be the transition word of the run \(\rho_{i,j}\). Define \(Flat_i:=\to_{E(\rho_{i,1})}^{\ast} \circ \dots \circ \to_{E(\rho_{i, r_i})}^{\ast}\). By induction, \(\vectSet{R}_i \subseteq Flat_i':=\to_{w_{i,1}}^{\ast} \circ \dots \circ \to_{w_{i,s_i}}^{\ast}\) for some transition words \(w_{i,1}, \dots ,w_{i,s_i} \in E^{\ast}\). Then \(\vectSet{R} \subseteq Flat_0 \circ Flat_1' \circ \dots \circ Flat_r' \circ Flat_r\) is transition word flattable w.r.t. \(\VAS\).
\end{proof}

\LemmaExchangePlusAndComposition*

\begin{proof}
Let \((\vect{v}_i, \vect{x}_i) \in \vectSet{R}_i, (\vect{x}_i, \vect{w}_i) \in \vectSet{R}_i'\) for \(i \in \{1,\dots,r\}\). \((\sum_{i=1}^r \vect{v}_i, \sum_{i=1}^r \vect{x}_i) \in \sum_{i=1}^r \vectSet{R}_i\) and \((\sum_{i=1}^r \vect{x}_i, \sum_{i=1}^r \vect{w}_i) \in \sum_{i=1}^r \vectSet{R}_i'\) imply \(\sum_{i=1}^r (\vect{v}_i, \vect{w}_i) \in (\sum_{i=1}^r \vectSet{R}_i) \circ (\sum_{i=1}^r \vectSet{R}_i')\).
\end{proof}

\LemmaIntersectionSemilinear*

\begin{proof}
W.l.o.g. \(\vectSet{S}\) is linear, since PVAS sections are clearly closed under union. Let \(\vectSet{X}\) be defined by a \(d\)-dimensional PVAS \(\VAS\). We give a \(6d\)-dimensional PVASS \(\VAS'\) defining \(\vectSet{X} \cap \vectSet{S}\) as follows. It has four states: In the first state, it writes \((\vect{x}, \vect{x})\) on its first \(2d\) counters for any vector \(\vect{x} \in \N^d\). Then it non-deterministically moves to the next state, where it applies \(\VAS\) on the first \(d\) counters (we would want to use the second \(d\) counters because this would express exactly \(\rightarrow_{\VAS}^{\ast}\), but zero tests have to be on the frontmost counters). Then it non-deterministically moves to state 3 while adding the base of \(\vectSet{S}\) to the middle \(2d\) counters. In state 3 it adds any number of periods of \(\vectSet{S}\). It then non-deterministically moves to state 4. In state 4 it checks that the first \(2d\) and the middle \(2d\) counters coincide, while copying them to the last \(2d\) counters. Usually such a check is done using transitions \((-\vect{e}_i, -\vect{e}_i, +\vect{e}_i)\) for unit vectors \(\vect{e}_i \in \N^{2d}\), but here we have to be careful of possible necessary permutations (like the first \(2d\) counters being ordered the wrong way).
\end{proof}

\subsection{Proofs of Section \ref{SubsectionEStarSmooth}} \label{SubsectionAppendixTStarCase}

Fix a pair \(\gamma=(\Expression^{\ast}, \vect{c}, P)\), where \(P \subseteq \vectSet{P}_{\vect{c}}\) is a periodic relation. For convenience, we repeat the definition of the graph \(G_{\gamma}\) and related objects here. Remember that the goal is to prove property \(2.'\) for \(\vectSet{P}_{\Expression^{\ast},\vect{c}}\), i.e. to determine a definable cone relation \(\vectSet{R}_{\gamma}\) containing \(\dir(P)\) and such that for every \((\vect{e}, \vect{f}) \in \vectSet{R}_{\gamma}\), there exists \(n \in \N\) and \((\vect{a}, \vect{b}) \in P\) (\(\in P\) being crucial) such that \(\dirOfRun(\rho)+(\vect{a},\vect{b})+\N (\vect{e}, \vect{f})\) is flattable w.r.t. \(\Expression\).

We denote by \(\Omega_{\gamma}\) the set of runs \(\rho \in \Omega(\Expression^{\ast})\) such that \(\dirOfRun(\rho) \in (\vect{c}, \vect{c})+P\). \(\Omega_{\gamma} \neq \emptyset\) since the run \((\vect{c}, \epsilon, \vect{c})\) not performing any transitions is in \(\Omega_{\gamma}\). The idea is to split counters into bounded and unbounded counters for \(\Omega_{\gamma}\). Hence let \(Q_{\gamma} \subseteq \N^{d'}\) be the set of configurations obtained along some run \(\rho \in \Omega_{\gamma}\). We denote by \(I_{\gamma}\) the set of indices such that \(\{\vect{q}(i) \mid \vect{q} \in \vectSet{Q}_{\gamma}\}\) is finite, i.e. the set of bounded counters. We consider the projection \(\pi_{\gamma} \colon \N^{d'} \to \N^{I_{\gamma}}\) to the bounded counters. We define a finite directed multigraph \(G_{\gamma}\) with vertices \(S_{\gamma}:=\pi_{\gamma}(\vectSet{Q}_{\gamma})\). Consider the set \(\Omega_{s,t}\) of generalized transitions \(\eta\) with \(\pi_{\gamma}(\source(\eta))=s\) and \(\pi_{\gamma}(\target(\eta))=t\), and such that \(\eta\) occurs in some run \(\rho \in \Omega_{\gamma}\). We add an edge \((s,t)\) for every minimal (w.r.t. \(\leq_{\Omega(\Expression)}\)) element of \(\Omega_{s,t}\). We let \(s_{\gamma}=\pi_{\gamma}(\vect{c})\) denote the ``initial state'' for this graph.

Remember the outline: Prove that runs \(\rho \in \Omega(\Expression^{\ast})\) correspond to cycles in \(G_{\gamma}\), and that unbounded counters are simultaneously unbounded. First we define ``pumping vectors''.

An \emph{intraproduction} for \(\gamma\) is a vector \(\vect{h} \in \N^{d'}\) such that \(\vect{c}+\vect{h} \in \vectSet{Q}_{\gamma}\). We denote by \(\vectSet{H}_{\gamma}\) the set of intraproductions for \(\gamma\). This set is periodic by the following lemma, which is proven the same way as in \cite{Leroux13}, but whose proof shows the power of monotonicity which the following proofs also build on.

\begin{lemma}
\(\vectSet{Q}_{\gamma}+\vectSet{H}_{\gamma} \subseteq \vectSet{Q}_{\gamma}\).
\end{lemma}

\begin{proof}
Let \(\vect{q} \in \vectSet{Q}_{\gamma}, \vect{h} \in \vectSet{H}_{\gamma}\). As \(\vect{q} \in \vectSet{Q}_{\gamma}\), there exists a pair \((\vect{x},\vect{y}) \in P\) such that \(\vect{c}+\vect{x} \to_{\Expression}^{\ast} \vect{q} \to_{\Expression}^{\ast} \vect{c}+\vect{y}\). Since \(\vect{h} \in \vectSet{H}_{\gamma}\) there exists a pair \((\vect{x}', \vect{y}') \in P\) such that \(\vect{c}+\vect{x}' \to_{\Expression}^{\ast} \vect{c}+\vect{h} \to_{\Expression}^{\ast} \vect{c}+\vect{y}'\). Since \(\to_{\Expression}^{\ast}\) is monotone, we obtain \(\vect{c}+(\vect{x}+\vect{x}') \to_{\Expression}^{\ast} \vect{q}+\vect{h} \to_{\Expression}^{\ast} \vect{c}+(\vect{y}+\vect{y}')\). Since \(P\) is periodic, we obtain \(\vect{q}+\vect{h} \in \vectSet{Q}_{\gamma}\).
\end{proof}

Importantly, this implies that if for some intraproduction \(\vect{h}\) and counter \(i\) we have \(\vect{h}(i)>0\), then \(\{\vect{q}(i) \mid \vect{q} \in Q_{\gamma}\}\) is unbounded. I.e., if \( i \in I_{\gamma}\), then \(\vect{h}(i)=0\).

\begin{corollary}
We have \(\pi_{\gamma}(\source(\rho))=s_{\gamma}=\pi_{\gamma}(\target(\rho))\) for every run \(\rho \in \Omega_{\gamma}\).
\end{corollary}

\begin{proof}
Since \(\rho \in \Omega_{\gamma}\) there exists \((\vect{x}, \vect{y})\in P\) such that \(\rho\) is a run from \(\vect{c}+\vect{x}\) to \(\vect{c}+\vect{y}\). Hence \(\vect{x}, \vect{y}\) are intraproductions, and \(\vect{x}(i)=0=\vect{y}(i)\) for every \(i \in I_{\gamma}\). Therefore \(\pi_{\gamma}(\source(\rho))=\pi_{\gamma}(\vect{c})=\pi_{\gamma}(\target(\rho))\).
\end{proof}

In particular, runs \(\rho \in \Omega_{\gamma}\) induce loops in \(G_{\gamma}\), and \(G_{\gamma}\) is hence strongly connected. 

Next, we want to prove that in fact every loop on \(s_{\gamma}\) is also induced by some run \(\rho \in \Omega_{\gamma}\), as well as the claim about simultaneous unboundedness of all unbounded counters. This requires multiple lemmas with again the same proof as in \cite{Leroux13}:

\begin{lemma}
\cite[Lemma VIII.6]{Leroux13} For every \(\vect{q} \leq \vect{q}'\) in \(\vectSet{Q}_{\gamma}\) there exists an intraproduction \(\vect{h} \in \vectSet{H}_{\gamma}\) such that \(\vect{q}' \leq \vect{q}+\vect{h}\).
\end{lemma}

\begin{lemma}
\cite[Lemma VIII.7]{Leroux13} There exists an intraproduction \(\vect{h} \in \vectSet{H}_{\gamma}\) such that \(\{i \mid \vect{h}(i)=0\}=I_{\gamma}\). \label{LemmaSimultaneouslyUnbounded}
\end{lemma}

Observe that Lemma \ref{LemmaSimultaneouslyUnbounded} proves the fact about ``simultaneous unboundedness'': By repeating such an intraproduction every unbounded counter increases arbitrarily \emph{simultaneously}. 

The following lemma shows that for every cycle in \(S_{\gamma}\), the lemma even states it for any sequence of states, one can find a run visitting them in sequence.

\begin{lemma}
\cite[Lemma VIII.9]{Leroux13} For every sequence \(s_1, \dots, s_k \in S_{\gamma}\) there exist \((\vect{x}, \vect{y}) \in P\) and \(\vect{q}_1, \dots, \vect{q}_k \in \vectSet{Q}_{\gamma}\) such that \(s_j=\pi_{\gamma}(\vect{q}_j)\) for every \(1\leq j \leq k\) and such that 

\(\vect{c}+\vect{x} \to_{\Expression}^{\ast} \vect{q}_1 \to_{\Expression}^{\ast} \dots \to_{\Expression}^{\ast} \vect{q}_k \to_{\Expression}^{\ast} \vect{c}+\vect{y}.\)\label{LemmaCyclesImplementable}
\end{lemma} 

Now we can start defining \(\vectSet{R}_{\gamma}\). The difficult (and new) part is how the \(\dir(\vectSet{P}_{\Expression, \eta_i})\) for the different generalized transitions \(\eta_i\) influence \(\vectSet{R}_{\gamma}\). The by far most important observation to answer this is that for every generalized transition \(\eta\), in particular for every edge in the graph \(G_{\gamma}\), the corresponding transformer relation \(\vectSet{P}_{\Expression, \eta}\) is \emph{reflexive}: It clearly contains \((\vect{0}, \vect{0})\), and since \(\Expression\) and hence \(\vectSet{P}_{\Expression, \eta}\) is monotone, it hence contains \((\vect{x}, \vect{x})\) for every \(\vect{x}\). Reflexivity is useful because if \(\vectSet{R}_i\) is reflexive for every \(i\), it implies \(\vectSet{R}_i \subseteq \vectSet{R}_1 \circ \dots \circ \vectSet{R}_k\). I.e. if some edge can pump a vector \((\vect{e}, \vect{f})\), then also some cycle can. And if some cycle can, then also a composition of minimal cycles can, which is non-trivial since not every cycle is a composition of minimal cycles (cycles can be inserted in the middle of a minimal cycle). So we can truly limit ourselves to single transitions.

Furthermore, \(\vectSet{P}_{\Expression, \eta}\) is smooth periodic fulfilling properties 1 and 2' since it can be written as composition of \(\vectSet{P}_{\Expression_i^{\ast}, \vect{c}_i}\) by Lemma \ref{LemmaEquivalentCharacterizationOfPRho}. Every one of these \(\vectSet{P}_{\Expression_i^{\ast}, \vect{c}_i}\) is smooth by induction, and the composition is hence smooth by Lemma \ref{LemmaCompositionSmooth}. Regarding (directions of) reflexive smooth periodic relations, we have

\begin{theorem}
\cite[Theorem VII.1]{Leroux13} Transitive closures of finite unions of reflexive definable cone relations over \(\Q_{\geq 0}^d\) are reflexive definable cone relations.\label{TheoremTransitiveClosureDefinable}
\end{theorem}

This leaves one last problem for defining \(R_{\gamma}\): The above suggests taking \(\dir(\vectSet{P}_{\Expression, \eta})\) for every edge \(\eta\), but this runs into a problem: Let \(\eta\) be an edge from state \(s\) to state \(t\). For \((\vect{e}, \vect{f}) \in \dir(\vectSet{P}_{\Expression, \eta})\), the larger run \(\eta' \geq_{\Omega(\Expression)} \eta\) where \((\vect{e}, \vect{f})\) is pumpable might fulfill \(\eta' \not \in \Omega_{s,t}\). This should remind the reader of a problem we discussed (and resolved) when dealing with composition of periodic relations in Section \ref{SectionSmoothnessComposition}: The answer is property 2'. 

Hence we define for every edge \(\eta\) the well-directed periodic relation \(P_{\eta}:=\{\dirOfRun(\eta')-\dirOfRun(\eta) \mid \eta' \geq_{\Omega(\Expression)} \eta, \eta' \in \Omega_{s,t}\} \subseteq \vectSet{P}_{\Expression, \eta}\). By property 2' for \(\vectSet{P}_{\Expression, \eta}\) we obtain a definable cone relation \(R_{\eta}\) containing \(\dir(P_{\eta})\) and such that for every \((\vect{e}, \vect{f}) \in R_{\eta}\) there exists a run \(\eta' \geq_{\Omega_(\Expression)} \eta\) with \(\eta' \in \Omega_{s,t}\) such that in \(\eta'\) we can pump \((\vect{e}, \vect{f})\), resolving our above problem.

\begin{remark}
There is another perspective on \(R_{\eta}\) and its definition for readers familiar with KLM decomposition. Similar to KLM decomposition, we require every object in the graph to be unbounded, in particular transitions. The problem we ran into is that pumps \((\vect{e}, \vect{f}) \in \dir(\vectSet{P}_{\Expression, \eta})\) might be bounded, even if their corresponding transition \(\eta\) is unbounded. The definition of \(R_{\eta}\) indirectly removes pumps \(\in \dir(\vectSet{P}_{\Expression, \eta})\) which are bounded.
\end{remark}

\begin{definition}
 For every cycle \(\rho_{cyc}=\eta_1 \dots \eta_k\) in \(G_{\gamma}\), let \(\dirOfRun(\rho_{cyc})\) be the basic effect of the cycle. We define \(\vectSet{R}_{\gamma,1}\) as the set of pairs \((\vect{e}, \vect{f})\) such that \(\vect{f}-\vect{e}\) is the basic effect of some cycle in \(G_{\gamma}\). This was used as \(\vectSet{R}_{\gamma}\) in the case of normal VASS, and is clearly definable.

We define \(\vectSet{R}_{\gamma,2}:=\bigcup_{\eta_{base}} R_{\eta_{base}}\), where \(\eta_{base}\) ranges over edges of \(G_{\gamma}\).

We define \(\vectSet{R}_{\gamma}:=(\vectSet{R}_{\gamma,1} \cup \vectSet{R}_{\gamma,2})^{\ast}\), which is definable by Theorem \ref{TheoremTransitiveClosureDefinable}.
\end{definition}

We have to prove that with this choice for \(\vectSet{R}_{\gamma}\), property 2' is fulfilled. I.e. we need \(\dir(P) \subseteq \vectSet{R}_{\gamma}\) as well as the following lemma.

\begin{lemma}
For every \((\vect{e},\vect{f}) \in \vectSet{R}_{\gamma}\) there exists \((\vect{a}, \vect{b}) \in P\) and \(n\in \N\) such that \((\vect{c}, \vect{c})+(\vect{a}, \vect{b})+\N n(\vect{e}, \vect{f})\) is flattable.
\end{lemma}

\begin{proof}
Let \((\vect{e}, \vect{f}) \in \vectSet{R}_{\gamma}\). Then there exists \(k \in \N\) and \(\vect{e}_0, \vect{e}_1, \dots, \vect{e}_k\) with \(\vect{e}_0=\vect{e}\) and \(\vect{e}_k=\vect{f}\) such that every \((\vect{e}_i, \vect{e}_{i+1}) \in \vectSet{R}_{\gamma,1} \cup \vectSet{R}_{\gamma,2}\). The idea is to create for every \(i\) a run \(\rho_i: \vect{c}+\vect{a}_i \to_{\Expression}^{\ast} \vect{c}+\vect{b}_i\) with \((\vect{a}_i, \vect{b}_i) \in P\) such that there exists \(\rho_i' \geq_{\Omega(\Expression^{\ast})} \rho_i\) fulfilling \(\dirOfRun(\rho_i')-\dirOfRun(\rho_i)=n(\vect{e}_i, \vect{e}_{i+1})\) for some \(n\). This would imply that \(\dirOfRun(\rho_i)+\N (\dirOfRun(\rho_i')-\dirOfRun(\rho_i)) =(\vect{c}, \vect{c})+(\vect{a}_i, \vect{b}_i)+\N n (\vect{e}_i, \vect{e}_{i+1})\) is flattable by Lemma \ref{LemmaBasicflattabilityProperties}(1.), and thereby the statement by taking the run \(\vect{c}+\sum_{i=0}^{k-1} \vect{a}_i \to_{\rho_0, \dots, \rho_{k-1}} \vect{c}+\sum_{i=0}^{k-1} \vect{b}_i\). 

In other words, without loss of generality \((\vect{e}, \vect{f}) \in \vectSet{R}_{\gamma,1} \cup \vectSet{R}_{\gamma,2}\).

The case of \(\vectSet{R}_{\gamma,1}\) is the same as in \cite[Lemma VIII.10]{Leroux13}. Hence we assume \((\vect{e}, \vect{f}) \in \vectSet{R}_{\gamma,2}\).

By definition of \(\vectSet{R}_{\gamma,2}\), there exist \((\vect{a}', \vect{b}') \in P_{\eta}\) and \(n \in \N\) such that \((\vect{a}', \vect{b}')+\N n (\vect{e}, \vect{f}) \subseteq \vectSet{P}_{\Expression, \eta}\). By definition of \(P_{\eta}\) there exists a run \(\eta' \in \Omega_{s,t}\) such that \(\eta' \geq_{\Omega(\Expression)} \eta\). The above containment then shows \(n(\vect{e}, \vect{f}) \in \vectSet{P}_{\eta'}\). By definition of \(\Omega_{s,t}\), there exists a run \(\rho \in \Omega_{\gamma}\) from \(\vect{c}+\vect{a} \to_{\rho} \vect{c}+ \vect{b}\) for some vector \((\vect{a}, \vect{b}) \in P\) such that \(\rho\) contains the step \(\eta'\). Since the relations \(\vectSet{P}_{\Expression, \eta}\) of the other transitions \(\eta\) on this run are reflexive, vectors pumpable into \(\eta'\) are pumpable into \(\rho\), i.e. we obtain \(n(\vect{e}, \vect{f}) \in \vectSet{P}_{\Expression, \eta'} \subseteq \vectSet{P}_{\Expression^{\ast}, \rho}\). In fact this containment is obtained via pumping a difference between runs and hence flattable by Lemma \ref{LemmaBasicflattabilityProperties} (1.).
\end{proof}

Next we prove \(\dir(P) \subseteq \vectSet{R}_{\gamma}\). We need the following, which works the same as in \cite{Leroux13}.

\begin{lemma}
\cite[Lemma VIII.11]{Leroux13} States in \(S_{\gamma}\) are incomparable. \label{LemmaStatesIncomparable}
\end{lemma}

%We also need a new lemma: When increasing \(\eta\) to \(\eta'\), only unbounded coordinates are changed.
%
%\begin{lemma}
%Let \(\eta \leq_{\Omega(\Expression)} \eta'\) in \(\Omega_{s,t}\). Then \((\vect{v}, \vect{w}):=\dir(\eta')-\dir(\eta)\) fulfills \(\vect{v}(i)=0\) and \(\vect{w}(i)=0\) for all coordinates \(i \in I_{\gamma}\). \label{LemmaInnerTransitionsOnlyAdaptUnboundedCoordinates}
%\end{lemma}
%
%\begin{proof}
%Part 1: Input counters: Since both \(\eta,\eta' \in \Omega_{s,t}\), in particular \(\pi_{\gamma}(\source(\eta'))=s=\pi_{\gamma}(\source(\eta))\). Hence \([\vect{v}](i)=0\) for all \(i \in I_{\gamma}\), since \(\pi_{\gamma}\) only projects away coordinates \(i \not \in I_{\gamma}\).
%
%Part 2: Output counters: Similarly, \(\eta, \eta' \in \Omega_{s,t}\) implies \(\pi_{\gamma}(\target(\eta))=t=\pi_{\gamma}(\target(\eta))\). 
%\end{proof}

Now we can prove that \(\dir(P) \subseteq \vectSet{R}_{\gamma}\).

\begin{lemma}
\(\dir(P) \subseteq \vectSet{R}_{\gamma}\).
\end{lemma}

\begin{proof}
Let \((\vect{e}, \vect{f}) \in \dir(P)\). By potentially rescaling, we obtain \((\vect{x}, \vect{y})\in P\) such that \((\vect{x}, \vect{y})+\N (\vect{e}, \vect{f}) \subseteq P\). For every \(n\), let \(\rho_n\) be a run with \(\dirOfRun(\rho_n)=(\vect{c}, \vect{c})+(\vect{x}, \vect{y})+n(\vect{e}, \vect{f})\). Since \(\leq_{\Omega(\Expression^{\ast})}\) is a wqo, there exist indices \(n,m\) such that \(\rho_n \leq_{\Omega(\Expression^{\ast})} \rho_m\). Write \(\rho_n=\eta_1 \dots \eta_k\). Since \(\rho_n \leq_{\Omega(\Expression^{\ast})} \rho_m\), we can write \(\rho_m=\rho_0' \eta_1' \rho_1' \dots \eta_k' \rho_k'\). Clearly we have \(\dirOfRun(\eta_i')-\dirOfRun(\eta_i) \in \vectSet{R}_{\gamma,2}\), since the run \(\rho_m \in \Omega_{\gamma}\) and hence \(\dirOfRun(\eta_i')-\dirOfRun(\eta_i) \in P_{\eta_i} \subseteq R_{\eta_{base,i}}\), where \(\eta_{base,i}\) is a minimal transition with \(\eta_{base,i} \leq_{\Omega(\Expression)} \eta_i\). So we mainly have to deal with the \(\rho_i'\). Let \(\vect{v}_i:=\source(\rho_i')\) for \(i \in \{0,\dots, k\}\), and \(\vect{w}_i:=\target(\rho_i')\). We claim that the \(\rho_i'\) are cycles.

Proof of claim: Since \(\target(\eta_i)=\source(\eta_{i+1}) \leq \source(\eta_{i+1}')\), we have \(s_i:=\pi_{\gamma}(\target(\eta_i)) \leq \pi_{\gamma}(\source(\eta_{i+1}'))\). By Lemma \ref{LemmaStatesIncomparable}, we obtain \(\pi_{\gamma}(\target(\eta_i))=\pi_{\gamma}(\source(\eta_{i+1}'))\). Similarly, from \(\target(\eta_i) \leq \target(\eta_i')\) and again Lemma \ref{LemmaStatesIncomparable} we obtain \(\pi_{\gamma}(\target(\eta_i))=\pi_{\gamma}(\target(\eta_i'))\). Hence \(\target(\eta_i')=\source(\rho_i')\) and \(\target(\rho_i')=\source(\eta_{i+1}')\) project to the same state \(s_i\). It follows that \(\rho_i'\) is a cycle on \(s_i\).

We claim that the effect of any cycle \(\rho'\) which occurs along a run \(\rho \in \Omega_{\gamma}\) is in \(\vectSet{R}_{\gamma}\).

Proof of claim: Let \(\rho'\) be a cycle which occurs along a run \(\rho\). Write \(\rho'=\eta_1' \dots \eta_k'\). Let \(\eta_i \leq \eta_i'\) be the corresponding minimal transitions (labels of the edges of \(G_{\gamma}\)). By definition we have that the basic effect of \(\eta_1 \dots \eta_k\) is in \(R_{\gamma,1}\). It remains to deal with the non-basic effect. Since \(\eta_i \leq_{\Omega(\Expression)} \eta_i'\) and the transition \(\eta_i'\) occurs along a run in \(\Omega_{\gamma}\), namely \(\rho\), we have \(\dirOfRun(\eta_i')-\dirOfRun(\eta_i) \in P_{\eta_i} \subseteq R_{\eta_i}\). Hence the extra effect \(\dirOfRun(\eta_1' \dots \eta_k')-\dirOfRun(\eta_1 \dots \eta_k)=[\dirOfRun(\eta_1')-\dirOfRun(\eta_1)] \circ \dots \circ [\dirOfRun(\eta_k')-\dirOfRun(\eta_k)] \in R_{\eta_1} \circ \dots \circ R_{\eta_k} \subseteq \vectSet{R}_{\gamma}\).
\end{proof}

\end{document}